\documentclass[10pt,conference]{IEEEtran}
\usepackage{cite}
\usepackage{amsmath,amssymb,amsfonts}
\usepackage{algorithmic}
\usepackage{graphicx}
\usepackage{textcomp}
\usepackage[table, dvipsnames]{xcolor} % custom
\usepackage[hyphens]{url}
\usepackage{fancyhdr}

\usepackage{amsthm}
\usepackage{xspace, xfrac}
\usepackage{enumitem}
\usepackage[caption=false]{subfig}
\usepackage[all]{nowidow}
\usepackage{comment}
\usepackage{booktabs, tabularx, multirow}
\usepackage{balance}

\usepackage{hyperref}
\usepackage[absolute,showboxes]{textpos}

\newcolumntype{L}{>{\raggedright\arraybackslash}X}
\newcolumntype{C}{>{\centering\arraybackslash}X}
\newcolumntype{R}{>{\raggedleft\arraybackslash}X}

\newcommand{\tablecenter}[1]{\multicolumn{1}{c}{{#1}}}

\newtheorem{theorem}{Theorem}
\newtheorem*{theorem*}{Theorem}
\newtheorem{definition}{Definition}
\newcommand{\eg}{e.g.\@,\xspace}

\newcommand{\NAME}{WHET\xspace}
\newcommand{\BASE}{SHARP8+\xspace}
\newcommand{\ctxt}[1]{{[\!\langle{#1}\rangle\!]}}
\newcommand{\ptxt}[1]{{\langle{#1}\rangle}}
\newcommand{\ring}{{\mathcal{R}_Q}}
\newcommand{\evk}[1]{{\mathsf{evk}_{#1}}}
\newcommand{\efflevel}{{L_\mathrm{eff}}}
\newcommand{\boot}{{\mathtt{Boot}}}

\newcommand{\textevk}{$\evk{}$\xspace}
\newcommand{\textevks}{$\evk{}$s\xspace}
\newcommand{\textefflevel}{$\efflevel{}$\xspace}
\newcommand{\textboot}{$\boot{}$\xspace}
\newcommand{\revb}[1]{\textcolor{black}{#1}}

\newcommand{\hpcayear}{2027}

\title{WHET: Welding Homomorphic Encryption to Accelerator Architectures}
\def\hpcacameraready{} % Uncomment to build camera-ready version

\newcommand\hpcaauthors{Jongmin Kim\IEEEauthorrefmark{2}\IEEEauthorrefmark{1}, Hyesung Ji\IEEEauthorrefmark{3}\IEEEauthorrefmark{1}, Wonseok Choi\IEEEauthorrefmark{3}, Hyunah Yu\IEEEauthorrefmark{3}, and {Jung Ho} Ahn\IEEEauthorrefmark{3}}
\newcommand\hpcaaffiliation{\IEEEauthorrefmark{2}Google, \IEEEauthorrefmark{3}Seoul National University}
\newcommand\hpcaemail{kjongmin@google.com, \{kevin5188, wonseok.choi, yhyuna, gajh\}@snu.ac.kr}

\author{
\IEEEauthorblockN{\hpcaauthors{}}
      \IEEEauthorblockA{
        \hpcaaffiliation{} \\
        \hpcaemail{} \\
      }
\vspace{-0.2in}
}

\author{
  \ifdefined\hpcacameraready
    \IEEEauthorblockN{\hpcaauthors{}}
      \IEEEauthorblockA{
        \hpcaaffiliation{} \\
        \hpcaemail{}
      }
  \else
    \IEEEauthorblockN{\normalsize{HPCA \hpcayear{} Submission
      \textbf{\#\hpcasubmissionnumber{}}} \\
      \IEEEauthorblockA{
        Confidential Draft \\
        Do NOT Distribute!!
      }
    }
  \fi 
}

\fancypagestyle{camerareadyfirstpage}{%
  \fancyhead{}
  
  \fancyhead[C]{
    \ifdefined\aeopen
    \parbox[][12mm][t]{13.5cm}{\hpcayear{} IEEE International Symposium on High-Performance Computer Architecture (HPCA)}    
    \else
      \ifdefined\aereviewed
      \parbox[][12mm][t]{13.5cm}{\hpcayear{} IEEE International Symposium on High-Performance Computer Architecture (HPCA)}
      \else
      \ifdefined\aereproduced
      \parbox[][12mm][t]{13.5cm}{\hpcayear{} IEEE International Symposium on High-Performance Computer Architecture (HPCA)}
      \else
      \parbox[][0mm][t]{13.5cm}{\hpcayear{} IEEE International Symposium on High-Performance Computer Architecture (HPCA)}
    \fi 
    \fi 
    \fi 
    \ifdefined\aeopen 
      \includegraphics[width=12mm,height=12mm]{ae-badges/open-research-objects.pdf}
    \fi 
    \ifdefined\aereviewed
      \includegraphics[width=12mm,height=12mm]{ae-badges/research-objects-reviewed.pdf}
    \fi 
    \ifdefined\aereproduced
      \includegraphics[width=12mm,height=12mm]{ae-badges/results-reproduced.pdf}
    \fi
  }
  \fancyfoot[C]{}
}
\TPMargin{5pt}

\newcommand{\copyrightstatement}{
    \begin{textblock}{0.84}(0.08,0.93)    % tweak here: {box width}(leftposition, rightposition)
         \noindent
         \footnotesize
         \textcopyright\ 2026 IEEE. Personal use of this material is permitted. Permission from IEEE must be obtained for all other uses, in any current or future media, including reprinting/republishing this material for advertising or promotional purposes, creating new collective works, for resale or redistribution to servers or lists, or reuse of any copyrighted component of this work in other works.
    \end{textblock}
}

\begin{document}
\copyrightstatement
\maketitle

%Enables the camera ready header and footer
\ifdefined\hpcacameraready 
  \pagestyle{plain}
\else
  \thispagestyle{plain}
  \pagestyle{plain}
\fi

\newcommand{\hpcaheight}{0mm}
\ifdefined\eaopen
\renewcommand{\hpcaheight}{12mm}
\fi

{\renewcommand{\thefootnote}{}\footnotetext{{}\IEEEauthorrefmark{1} Both authors contributed equally to this research. Work done while Jongmin Kim was at Seoul National University.}
\pagenumbering{arabic}
%%%%%%%%%%%%%%%%%%%%%%%%%%%%%%%%%%%%%%%%
%%%%%%%% -- PAPER CONTENT STARTS -- %%%%%%%%%

\begin{abstract}

Fully homomorphic encryption (FHE) enables computations on encrypted data without decryption, offering strong data privacy at the expense of substantial computational and memory overheads.
Prior efforts have steadily improved FHE performance through cryptographic and algorithmic enhancements or hardware acceleration, yet these two directions have progressed largely in isolation, hindering the full exploitation of available hardware capabilities.

This work presents WHET, which introduces memory-centric, architecture-aware optimizations to better align cryptographic and algorithmic constructions with FHE accelerators.
We first identify conventional FHE constructions as major sources of excessive working sets and heavy off-chip memory traffic.
We propose accelerator-specific algorithmic techniques, including fine-grained coefficient-to-slot transformation, plaintext compression, and intermediate modulus raising, 
to reduce the on-chip data footprint by minimizing temporary ciphertexts and plaintext loads.
With these techniques applied, we unveil additional opportunities to improve on-chip memory efficiency; we introduce lightweight architectural refinements, including a special-purpose buffer and functional unit extensions.
With these optimizations, WHET achieves 1.38--8.74$\times$ per-area performance improvements over state-of-the-art FHE accelerators and the first-ever sub-millisecond CKKS bootstrapping.

\end{abstract}

\section{Introduction}
\label{sec:intro}

Homomorphic encryption (HE~\cite{2021-standard}) enables computations on encrypted data without decrypting them, providing true end-to-end data privacy.
HE holds great potential for secure cloud computing, privacy-preserving machine learning (ML), and confidential data analytics, as exemplified by Apple's HE-based private visual search service~\cite{apple-private-search}.
While early HE schemes supported only a limited number of operations, fully HE (FHE) schemes have ushered in a new era by introducing bootstrapping (\textboot), which enables unlimited computations on encrypted data~\cite{stoc-2009-gentry-fhe}.
This work focuses on CKKS~\cite{asia-2017-ckks}, a popular FHE scheme.
CKKS is well-suited for ML tasks as it efficiently supports arithmetic on complex or real numbers.

Despite its potential, FHE induces substantial computational and memory overheads, making current FHE applications prohibitively expensive for practical use.
For instance, convolutional neural network (CNN) inference that takes a fraction of a second on CPUs can take as long as 1,447 seconds~\cite{asplos-2025-orion} under FHE---over 30,000$\times$ slower (see Table~\ref{tab:eval-orion}).

Numerous studies have explored FHE acceleration on readily available hardware, including CPUs~\cite{wahc-2021-hexl, access-2021-demystify}, GPUs~\cite{asplos-2026-cheddar, hpca-2023-tensorfhe, hpca-2025-warpdrive, isca-2026-hyperdrive,tches-2021-100x}, and FPGAs~\cite{hpca-2023-fab, hpca-2023-poseidon, hpca-2025-effact}.
However, conventional hardware designs are not built to maximize the efficiency of the modular integer arithmetic required by FHE.
More critically, due to the limited on-chip memory capacity and bandwidth, they become memory-bound~\cite{hpca-2025-anaheim}.

\begin{figure}
    \centering
    \subfloat[Level vs. computational cost]{\includegraphics[width=0.478\columnwidth]{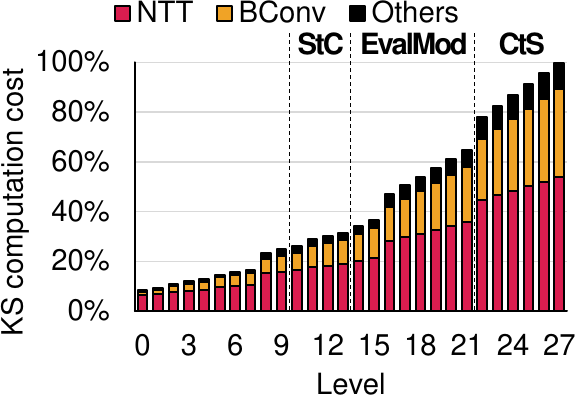}\label{fig:level:compute}}
    \;
    \subfloat[Level vs. data size]{\includegraphics[width=0.478\columnwidth]{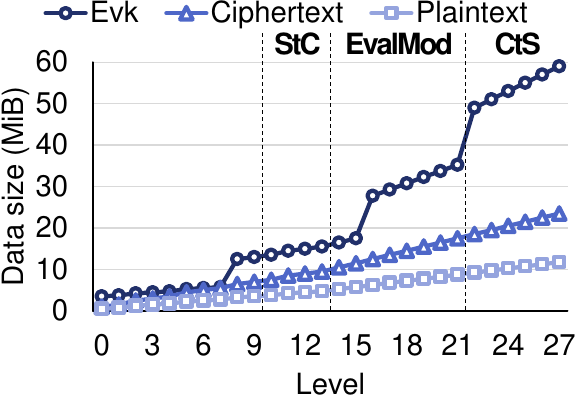}\label{fig:level:evk}}
    \caption{(a) $\mathtt{KS}$ computational cost breakdown and (b) data size by level. Computational costs are weighted sums of integer multiplication and modular reduction counts~\cite{wahc-2023-32bit}. A random seed replaces half of an \textevk~\cite{isca-2022-craterlake}. Higher levels are reserved for CtS, EvalMod, and StC, which comprise \textboot.
    }
    \label{fig:level}
\end{figure}

The limitations of existing hardware have motivated numerous custom ASIC accelerator proposals~\cite{isca-2022-bts, micro-2022-ark, isca-2023-sharp, isca-2022-craterlake, isca-2025-fast, micro-2025-hawk, micro-2024-trinity, micro-2024-ufc, asplos-2026-falcon, hpca-2026-unifhe}.
While their designs differ in detail, they share several common structural characteristics.
First, they aim to maximize the throughput of number-theoretic transform (NTT), which accounts for over half of the total computational cost (see Fig.~\ref{fig:level:compute}).
Second, to mitigate memory bottlenecks in conventional hardware, they dedicate nearly half of their chip area to large on-chip memory (\textgreater200MiB), which significantly restricts architectural efficiency and flexibility.

In this paper, \emph{we question whether existing accelerators effectively capitalize on their substantial hardware resources.}
Our analysis (\S\ref{sec:motivation}) reveals that large working sets stemming from conventional cryptographic and algorithmic constructions limit effective utilization of accelerator hardware.
We focus on the coefficient-to-slot transformation (CtS), a major contributor to FHE runtime.
CtS operates at the highest levels where data object sizes are largest (see Fig.~\ref{fig:level:evk}), whereas lower levels are largely free from working-set constraints.
Well-known techniques such as the baby-step giant-step (BSGS) algorithm~\cite{crypto-2018-linear, eurocrypt-2021-efficient} and CtS matrix decomposition~\cite{access-2019-dft} reduce the computational complexity of CtS.
However, their use in prior FHE accelerators largely follows cryptographic and algorithmic conventions without accounting for hardware characteristics. % capabilities.
Thus, their constructions inflate the CtS working set, forcing accelerators to either provision even larger on-chip memory or suffer from off-chip memory bandwidth bottlenecks, leaving both on-chip memory bandwidth and computational units severely underutilized.

Motivated by these challenges, we introduce \textbf{\NAME}, a set of cryptographic and algorithmic techniques tailored to accelerator execution (\S\ref{sec:alg}).
First, we develop fine-grained CtS (fg-CtS), which decomposes CtS into many levels and avoids the use of BSGS.
This creates only a small number of temporary ciphertexts (up to 24MiB each) and thereby significantly reduces the working set.
Second, we discover a repetitive structure in CtS plaintexts, exploiting it to compress the plaintexts by up to 8,192$\times$ \revb{in theory (32$\times$ in \NAME, \S\ref{sec:arch:compress})}.
Finally, we introduce intermediate ModRaise, which complements these two techniques by reducing the level consumption of fg-CtS while shrinking the size of non-compressible CtS plaintexts through modulus reduction.
Together, these techniques preserve limited on-chip memory capacity and significantly reduce off-chip memory traffic while maintaining 128-bit security guarantees~\cite{jmc-2025-lattice-estimator}.

As a part of \NAME, we also introduce two lightweight architectural modifications to improve on-chip resource utilization (\S\ref{sec:arch}).
These modifications are motivated by our analysis of accelerator performance limits under \NAME's cryptographic and algorithmic enhancements.
First, we add a special-purpose buffer that stores frequently accessed intermediate results produced by evaluation key multiplication (KeyMult) to mitigate memory bandwidth limitations.
Second, we introduce instruction extensions for the element-wise functional unit to alleviate data dependency stalls.
These enhancements also provide useful insights for future accelerator designs.

We evaluate \NAME on \BASE, built upon SHARP~\cite{isca-2023-sharp} with key improvements from recent architectural advances (\S\ref{sec:arch:base}).
\NAME achieves 1.84$\times$ mean speedup and 2.30$\times$ mean energy-delay-product improvement over the strong \BASE baseline while using smaller on-chip memory capacity and chip area.
Also, compared to state-of-the-art FHE ASIC accelerators, \NAME delivers 1.38--8.74$\times$ higher performance per area. 
Highlights of our results include the first sub-millisecond CKKS full-slot \textboot (refreshing $2^{15}$ complex numbers) and encrypted CNN inference latencies of 33.1--239ms, enabling real-time encrypted inference compared with the original Orion~\cite{asplos-2025-orion} implementation, which requires 301--1,447s on a single-threaded CPU (see Table~\ref{tab:eval-orion}).

The key contributions of \NAME are summarized below:
\begin{itemize}%[leftmargin=*]
    \item We propose \emph{fine-grained CtS (fg-CtS)}. By aggressively decomposing CtS and avoiding the baby-step giant-step algorithm, fg-CtS significantly reduces the CtS working set, enabling efficient accelerator execution.
    \item We introduce \emph{plaintext compression} and \emph{intermediate ModRaise} algorithms that reduce the size of plaintexts and the amount of off-chip memory traffic.
    \item We design architectural enhancements over the already strong SHARP~\cite{isca-2023-sharp} baseline, including a special-purpose buffer and functional unit instruction extensions, to improve on-chip resource utilization.
\end{itemize}

\section{FHE Basics}
\label{sec:back}

Homomorphic encryption (HE) enables servers to perform computations directly on clients' ciphertexts without decryption.
We introduce CKKS~\cite{asia-2017-ckks, sac-2018-frns-ckks, rsa-2020-better}, which offers the highest throughput among HE schemes~\cite{iacr-2023-demystify-boot} and supports fixed-point arithmetic, making it well-suited for numerous real-world tasks.
Notations and symbols are summarized in Table~\ref{tab:notation}.

\begin{table}[t]
\caption{Notations and symbols.}
\label{tab:notation}
\begin{tabularx}{0.99\columnwidth}{lX}
\toprule
Symb. & Explanation\\
\midrule
$\ring$ & Polynomial ring $\mathbb{Z}_Q[\mathcal{X}]/(X^N+1)$.\\
$N$ & Degree of $\ring$ (typically, $2^{16}$).\\
$Q$ & Modulus of $\ring$.\\
$\mathcal{Q}_i$ & $i$-th prime sub-modulus of $Q$.\\
$P$  & Additional modulus for \textevks.\\
$\mathcal{P}_i$ &  $i$-th prime sub-modulus of $P$.\\
\textefflevel & Effective level: the level after \textboot.\\
$a$ & Polynomial in $\ring$.\\
$a(\mathcal{X})$ & Polynomial $a$ with its variable $\mathcal{X}$ specified.\\
$\ptxt{\mathbf{u}}$ & Plaintext encoding a vector $\mathbf{u}\in\mathbb{C}^{N/2}$. $\ptxt{\mathbf{u}}\in\ring$.\\
$\ctxt{\mathbf{u}}$ & Ciphertext encrypting $\ptxt{\mathbf{u}}$. $\ctxt{\mathbf{u}}\in\mathcal{R}_Q^2$.\\
\textevk & Evaluation key. $\evk{} \in\mathcal{R}_{PQ}^{2\times \beta}$.\\
$\beta$ & \textevk decomposition number ($\mathtt{dnum}$ in \cite{rsa-2020-better}).\\
$\Delta$ & Scale of a ciphertext, affecting precision and \textefflevel.\\
$D_\mathrm{tr}$ & CtS/StC (Coefficient-to-Slot/Slot-to-Coefficient) matrix decomposition number.\\
\bottomrule
\end{tabularx}
\vspace{-0.1in}
\end{table}

\subsection{Basic CKKS function evaluations}
\label{sec:back:he}

CKKS~\cite{asia-2017-ckks} operates on polynomials, denoted in plain lowercase (\eg $a$ or $a(\mathcal{X})$ with the variable $\mathcal{X}$ specified).
A client encodes a complex vector $\mathbf{u} \in \mathbb{C}^{N/2}$ into a \emph{plaintext} polynomial $u=\ptxt{\mathbf{u}}$,
then encrypts it into a \emph{ciphertext} $\ctxt{\mathbf{u}}$, which comprises a pair of polynomials.

The following are server-side function evaluations of additions, element-wise multiplications ($\odot$), and cyclic rotations ($\ll$) on vectors encrypted as $\ctxt{\mathbf{u}}=(a,b)$ and $\ctxt{\mathbf{v}}=(a',b')$:
\begin{itemize}[leftmargin=*]
    \item $\mathtt{HAdd}$: $\ctxt{\mathbf{u} + \mathbf{v}} \gets (a+a', b+b')$
    \item $\mathtt{PMult}$: $\ctxt{\mathbf{p}\odot \mathbf{u}} \gets (p \cdot a, p \cdot b)$ for plaintext $p=\ptxt{\mathbf{p}}$
    \item $\mathtt{HMult}$: $\ctxt{\mathbf{u}\odot \mathbf{v}} \gets (a\cdot b' + a' \cdot b, b\cdot b')+\mathtt{KS}(a\cdot a', \evk{\odot})$
    \item $\mathtt{HRot}$: $\ctxt{\mathbf{u} \ll R} \gets (0, b(\mathcal{X}^{5^R})) + \mathtt{KS}(a(\mathcal{X}^{5^R}), \evk{R})$
\end{itemize}

Here, $\mathtt{KS}$ (key-switching) dominates CKKS computations, involving dozens of polynomial operations~\cite{tches-2021-100x}.
$\mathtt{KS}$ requires client-provided evaluation keys (\textevks), each comprising $2\times \beta$ polynomials (\eg $\beta = 4$)~\cite{rsa-2020-better}.
$\mathtt{HMult}$ requires an $\evk{\odot}$, and each $\mathtt{HRot}$ with a different distance $R$ requires a distinct $\evk{R}$.

\subsection{Polynomial operations}
\label{sec:back:poly-ops}

We focus on polynomial operations constituting CKKS function evaluations.
The main polynomial domain is a ring $\ring = \mathbb{Z}_Q[X]/(X^N+1)$.
A polynomial in $\ring$ has $N$ (typically $2^{16}$) integer coefficients modulo a huge $Q$, reaching $2^{1747}$ under standard 128-bit security constraints~\cite{iacr-2024-guideline}.

Modern HE uses the residue number system (RNS) to handle $Q$ efficiently~\cite{sac-2018-frns-ckks}.
$L$ small (\eg $<2^{32}$) prime sub-moduli $\mathcal{Q}_0,\cdots, \mathcal{Q}_{L-1}$ satisfying $Q=\prod_{i=0}^{L-1}\mathcal{Q}_i$ are chosen.
For a polynomial $a\in\ring$, we obtain $(a_{\mathcal{Q}_0}, \cdots, a_{\mathcal{Q}_{L-1}})$ by taking the residues modulo $\mathcal{Q}_i$ for each coefficient of $a$.

Each $a_{\mathcal{Q}_i}$ is referred to as a \emph{limb}~\cite{micro-2023-mad} of $a$.
A polynomial $a\in\mathcal{R}_Q$ with $L$ limbs forms an $L\times N$ matrix, where the $i$-th row corresponds to $a_{\mathcal{Q}_i}$.
With RNS, polynomial additions and multiplications are performed per limb, leveraging hardware-friendly small-integer (\eg int32) arithmetic.

There are four types of polynomial operations:

\begin{itemize}[leftmargin=*]
\item \textbf{Number-theoretic transform (NTT)}.
Multiplying $a_{\mathcal{Q}_i}$ with $b_{\mathcal{Q}_i}$ requires a negacyclic convolution of their coefficients ($\mathcal{O}(N^2)$).
NTT uses the fast Fourier transform (FFT)~\cite{cooley-tukey} ($\mathcal{O}(N\log N)$) to convert this into an element-wise multiplication ($\mathcal{O}(N)$): $\mathrm{NTT}(a_{\mathcal{Q}_i}\cdot b_{\mathcal{Q}_i})=\mathrm{NTT}(a_{\mathcal{Q}_i})\odot\mathrm{NTT}(b_{\mathcal{Q}_i})$.
We apply NTT per limb and keep polynomials in this ``NTT domain'' by default.
NTT dominates HE computations, accounting for over half the costs as shown in Fig.~\ref{fig:level:compute}.

\item \textbf{Base conversion (BConv)}.
BConv changes the polynomial modulus through a matrix-matrix multiplication~\cite{micro-2022-ark}.
A polynomial $a\in\mathcal{R}_Q$ with $L$ limbs ($L\times N$ matrix)
is multiplied
by an $L'\times L$ matrix of precomputed constants,
yielding $L'$ output limbs ($L'\times N$ matrix).

\item \textbf{Automorphism}.
Automorphism maps $a(\mathcal{X})$ to $a(\mathcal{X}^{5^R})$, inducing a permutation of the coefficients of $a$.

\item \textbf{Element-wise operations}.
The remaining operations, such as $a \pm b$ and $a \cdot b$ (excluding NTT), are inherently element-wise over the $L\times N$ elements in a polynomial.
\end{itemize}

\subsection{Bootstrapping: The defining feature of FHE}
\label{sec:back:boot}

High-throughput HE schemes are \emph{leveled}.
Applications start with a large modulus (\eg $Q_\mathrm{top}\simeq2^{1400}$), but multiplicative functions, such as $\mathtt{HMult}$ and $\mathtt{PMult}$, require modulus reductions to control error growth~\cite{asia-2017-ckks}.
Each reduced $Q$ corresponds to a unique level from $\mathrm{top}$ ($Q_\mathrm{top}$) to 0 (\eg $Q_\mathrm{bot}\simeq2^{50}$).
At level 0, we cannot perform more function evaluations.

Bootstrapping (\textboot) is the defining feature of fully HE (FHE), a class of HE schemes with unlimited function evaluation capabilities. 
\textboot restores the modulus (and the level) of a ciphertext at level 0 to continue function evaluations.
\textboot comprises four sequential steps in CKKS~\cite{eurocrypt-2018-heaanboot, eurocrypt-2021-efficient, acns-2022-sparseboot, rsa-2020-better}:
\begin{enumerate}[leftmargin=*]
    \item \textbf{ModRaise}. Raising the modulus to $Q_\mathrm{top}$ from $Q_\mathrm{bot}$. This introduces unwanted $Q_\mathrm{bot}$ multiples into the encrypted data, which are removed by subsequent \textboot steps.
    \item \textbf{CtS}. Coefficient-to-slot transformation, which performs an encrypted linear transform evaluation multiplying the $\mathrm{top}$-level ciphertext with an $\frac{N}{2}\times\frac{N}{2}$ matrix.
    \item \textbf{EvalMod}. Evaluating an approximate mod-$Q_\mathrm{bot}$ function, which primarily requires $\mathtt{HMult}$ evaluations.
    \item \textbf{StC}. Slot-to-coefficient transformation, which is another encrypted linear transform evaluation at lower levels.
\end{enumerate}

The number of remaining levels after \textboot is referred to as the effective level (\textefflevel)~\cite{isca-2023-sharp}.
Increasing \textefflevel may reduce the frequency of \textboot.

\section{The Working Set Problem in CtS}
\label{sec:motivation}

\subsection{Overview}
\label{sec:motivation:overview}

CtS at the highest levels is the most computationally demanding step of \textboot, accounting for 60.3\% of total \textboot time on a conventional accelerator (\BASE in Fig.~\ref{fig:sensitivity:boot}).
As \textboot is a vital FHE functionality, usually dominating workload execution times~\cite{micro-2022-ark, icml-2022-resnet, ccs-2024-neujeans}, inefficiencies in CtS represent a fundamental performance bottleneck.

We discover that \textbf{the excessive working set of conventional CtS becomes a critical hurdle for modern FHE accelerators}.
Prior work has hinted at this issue.
SHARP~\cite{isca-2023-sharp} observes that \textit{``we have to look out for the working set size only at high levels---that is, only while bootstrapping.''}
The working set grows superlinearly with the level, primarily due to the increase in \textevk size, which reaches up to 60MiB (see Fig.~\ref{fig:level:evk}).
Consequently, while lower levels have modest working sets, CtS requires hundreds of MBs.

This constraint forces architects to provision hundreds of MBs of on-chip memory to enable effective data reuse in CtS.
In turn, accelerators become heavily memory-dominated and limited in architectural flexibility.
For example, ARK~\cite{micro-2022-ark} allocates 54.8\% of its chip area to 588MiB of on-chip memory, SHARP~\cite{isca-2023-sharp} dedicates 48.9\% to approximately 200MiB, and the more recent FAST~\cite{isca-2025-fast} devotes 43.7\% to 281MiB.

Although the \textboot fraction varies across workloads (35.1--97.1\% in \S\ref{sec:eval}), this architectural burden imposed by CtS is fundamental.
Even workloads with modest \textboot proportions suffer substantial performance degradation (up to 1.91$\times$ slowdowns) when on-chip memory capacity is reduced from our 200MiB baseline (\BASE) to 104MiB, as the large CtS working set induces frequent off-chip memory accesses.

These observations motivate a deeper examination of conventional CtS configurations to identify the root causes of this inefficiency and uncover opportunities to overcome it.

\begin{figure*}[tb!]
    \centering
    \includegraphics[width=0.96\textwidth]{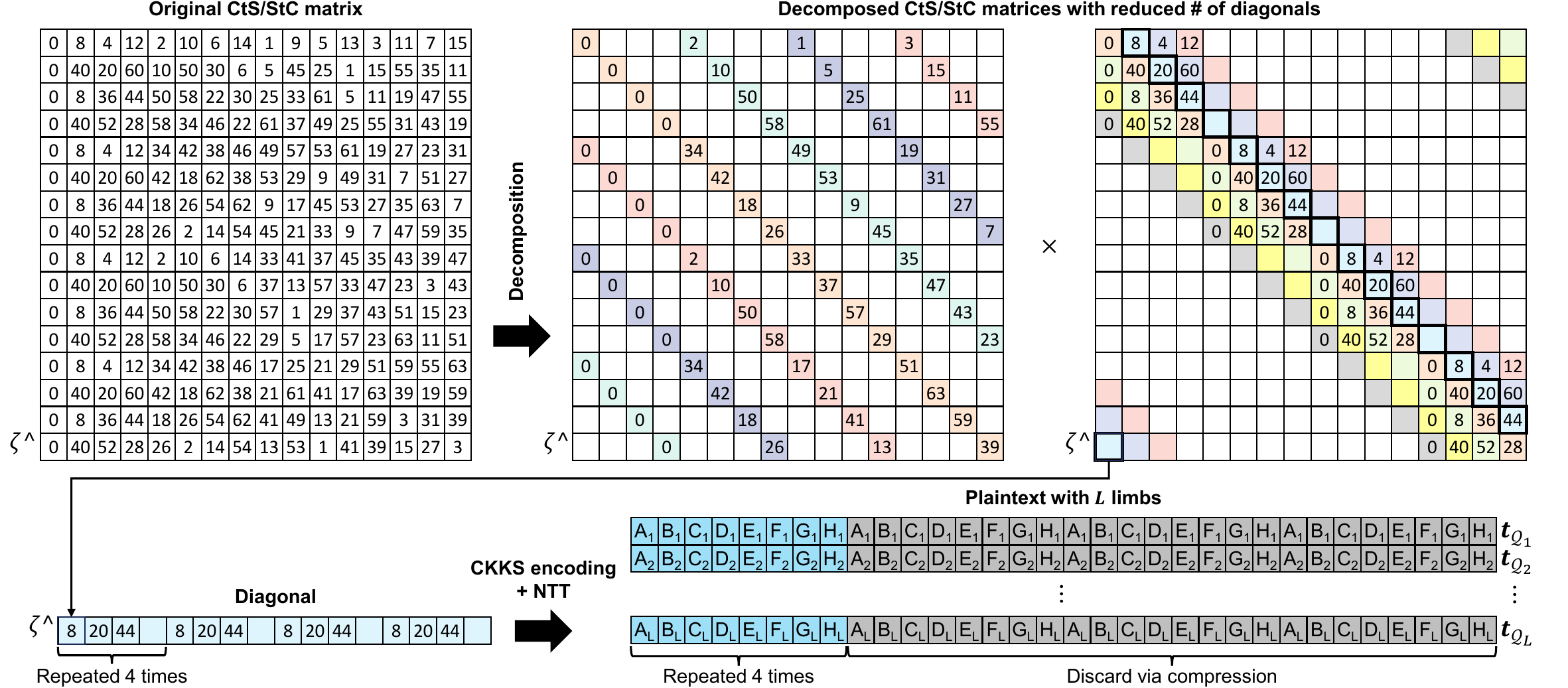}
    \vspace{-0.07in}
    \caption{CtS/StC matrix decomposition~\cite{access-2019-dft} and plaintext compression simplified for $N=32$. The original dense CtS/StC matrix, composed of powers of $\zeta=e^{\pi i/N}$, is decomposed into two sparse matrices with reduced \#diag (16 $\rightarrow$ 4 \& 7). Empty slots in each matrix represent zero values. The rightmost matrix contains diagonals with repetitions, allowing plaintext compression.}
    \label{fig:cts-optimizations}
    \vspace{-0.07in}
\end{figure*}

\subsection{Prior CtS optimizations}
\label{sec:motivation:lt}

We start from a simple linear transform example that evaluates
%We describe how linear transforms are performed in HE~\cite{crypto-2018-linear} using an example of
%evaluating
$\ctxt{\mathbf{M}\cdot\mathbf{u}}$ with a ciphertext $\ctxt{\mathbf{u}}$ and a $4\times4$ unencrypted matrix $\mathbf{M}$ shown below.
First, we extract the ``diagonals'' of $\mathbf{M}$ and encode them as plaintexts $\ptxt{{\color{red!80!black}\mathbf{d_0}}},\ptxt{{\color{green!80!black}\mathbf{d_1}}},\ptxt{{\color{blue!80!black}\mathbf{d_2}}},\ptxt{{\color{yellow!80!black}\mathbf{d_3}}}\in\ring$.
\[
\left[\begin{array}{cccc}
    \cellcolor{red!20}0 & \cellcolor{green!20}1 & \cellcolor{blue!20}2 & \cellcolor{yellow!20}3\\
    \cellcolor{yellow!20}4 & \cellcolor{red!20}5 & \cellcolor{green!20}6 & \cellcolor{blue!20}7\\
    \cellcolor{blue!20}8 & \cellcolor{yellow!20}9 & \cellcolor{red!20}10 & \cellcolor{green!20}11\\
    \cellcolor{green!20}12 & \cellcolor{blue!20}13 & \cellcolor{yellow!20}14 & \cellcolor{red!20}15
\end{array}\right]\rightarrow\left[\begin{array}{c}\cellcolor{red!20}0\\ \cellcolor{red!20}5\\ \cellcolor{red!20}10\\ \cellcolor{red!20}15\end{array}\right]\!\left[\begin{array}{c}\cellcolor{green!20}1\\ \cellcolor{green!20}6\\ \cellcolor{green!20}11\\ \cellcolor{green!20}12\end{array}\right]\!\left[\begin{array}{c}\cellcolor{blue!20}2\\ \cellcolor{blue!20}7\\ \cellcolor{blue!20}8\\ \cellcolor{blue!20}13\end{array}\right]\!\left[\begin{array}{c}\cellcolor{yellow!20}3\\ \cellcolor{yellow!20}4\\ \cellcolor{yellow!20}9\\ \cellcolor{yellow!20}14\end{array}\right]
\]
Then, we can use the property that
\begin{equation}
\label{eq:lt}
\mathbf{M}\cdot\mathbf{u} \gets {\mathbf{d_0}} \odot \mathbf{u} + {\mathbf{d_1}} \odot (\mathbf{u} \ll 1) + {\mathbf{d_2}} \odot (\mathbf{u} \ll 2) + {\mathbf{d_3}} \odot (\mathbf{u} \ll 3)
\end{equation}
to compute $\ctxt{\mathbf{M}\cdot\mathbf{u}}=\ptxt{\mathbf{d}_0}\times \ctxt{\mathbf{u}} + \ptxt{\mathbf{d}_1}\times \mathtt{HRot}(\ctxt{\mathbf{u}}, 1) + \ptxt{\mathbf{d}_2}\times \mathtt{HRot}(\ctxt{\mathbf{u}}, 2)+ \ptxt{\mathbf{d}_3}\times \mathtt{HRot}(\ctxt{\mathbf{u}}, 3)$, where $\times$ represents $\mathtt{PMult}$ and $+$ represents $\mathtt{HAdd}$.
For a general linear transform with $\mathrm{\#diag}$ non-zero diagonals in $\mathbf{M}$, we require $(\mathrm{\#diag}\!-\!1)$ $\mathtt{HRot}$s and $\mathrm{\#diag}$ $\mathtt{PMult}$s with plaintexts.

\noindent\textbf{$\bullet$ Minimum key-switching (Min-KS)}~\cite{micro-2022-ark, crypto-2018-linear}.
The above process requires $\mathtt{HRot}$s by distances $R=1, 2, 3$,
requiring three separate $\evk{R}$'s.
Instead, Min-KS obtains $\ctxt{\mathbf{u} \!\ll\! 2}$ by performing $\mathtt{HRot}$ by 1 on $\ctxt{\mathbf{u} \!\ll\! 1}=\mathtt{HRot}(\ctxt{\mathbf{u}}, 1)$.
Likewise, $\ctxt{\mathbf{u} \!\ll\! 3}$ is obtained from $\ctxt{\mathbf{u} \!\ll\! 2}$.
Thus, Min-KS enables evaluating the linear transform with only $\evk{1}$.

\noindent\textbf{$\bullet$ Baby-step giant-step algorithm (BSGS)}~\cite{crypto-2018-linear, eurocrypt-2021-efficient}.
Eq.~\ref{eq:lt} can be alternatively evaluated in two steps using precomputed $\mathbf{d}_2'=(\mathbf{d}_2 \gg 2), \mathbf{d}_3'=(\mathbf{d}_3 \gg 2)$:
\begin{equation}
\label{eq:bsgs}
\begin{split}
\text{BS:\ }& \mathbf{v} \gets (\mathbf{u} \ll 1)\\
\text{GS:\ }&\! \mathbf{M}\! \cdot\! \mathbf{u} \gets (\mathbf{d}_0 \!\odot\! \mathbf{u} + \mathbf{d}_1 \!\odot\! \mathbf{v}) \!+\! (\mathbf{d}_2' \!\odot\! \mathbf{u} + \mathbf{d}_3' \!\odot\! \mathbf{v}) \ll 2
\end{split}
\end{equation}
To generalize, the first baby-step (BS) performs iterative $\mathtt{HRot}$s by 1 (Min-KS) to produce $\sqrt{\mathrm{\#diag}}$ \emph{``baby'' ciphertexts}:
\[
\ctxt{\mathbf{u}}, \ctxt{\mathbf{u} \ll 1}, \cdots, \ctxt{\mathbf{u} \ll (\sqrt{\mathrm{\#diag}}-1)}.
\]
With these, the second giant-step (GS) generates $\sqrt{\mathrm{\#diag}}$ \emph{partial results} (\eg $\ctxt{\mathbf{d}_0 \!\odot \mathbf{u} + \mathbf{d}_1 \!\odot \mathbf{v}}$ and $\ctxt{\mathbf{d}_2' \!\odot \mathbf{u} + \mathbf{d}_3' \!\odot \mathbf{v}}$ for Eq.~\ref{eq:bsgs}), which are accumulated into a single ciphertext with $\mathtt{HRot}$s by a larger distance (another Min-KS).

Although BSGS requires two \textevks for a linear transform, BSGS significantly reduces the computational cost for large $\mathrm{\#diag}$ values.
Most computational cost of an encrypted linear transform evaluation stems from $\mathtt{KS}$ in $\mathtt{HRot}$ evaluations.
Counting the total number of $\mathtt{KS}$ ($\mathrm{\#}\mathtt{KS}$), BSGS reduces $\mathrm{\#}\mathtt{KS}$ from $(\mathrm{\#diag}-1)$ to $2(\sqrt{\mathrm{\#diag}}-1)$.

\noindent\textbf{$\bullet$ CtS/StC matrix decomposition}~\cite{access-2019-dft}.
The CtS/StC linear transform matrix can be decomposed into multiple matrices, each with a small number of non-zero diagonals.
Fig.~\ref{fig:cts-optimizations} shows an example for $N=32$, where a $16\times16$ matrix is decomposed into two matrices with four/seven $\mathrm{\#diag}$.
To generalize, a $\frac{N}{2} \times \frac{N}{2}$ CtS/StC matrix is decomposed into $D_{\mathrm{tr}}$ matrices, each with $\mathcal{O}(\log_{D_{\mathrm{tr}}} N)$ non-zero diagonals.
As the number of non-zero diagonals determines both the computation cost ($\mathrm{\#}\mathtt{KS}=\mathcal{O}(\sqrt{\mathrm{\#diag}})$) and the memory traffic ($\mathcal{O}(\mathrm{\#diag})$ plaintext loads), \emph{increasing $D_{\mathrm{tr}}$ is desirable} for large $N$ (e.g., $2^{16}$).
However, as a trade-off, \emph{this decomposition consumes $D_{\mathrm{tr}}$ levels during CtS}.

\subsection{Excessive working set in conventional CtS}
\label{sec:motivation:final}

\noindent \textbf{$\bullet$ Min-KS is essential for accelerators}.
State-of-the-art FHE accelerators~\cite{isca-2022-craterlake, micro-2022-ark, isca-2022-bts, isca-2023-sharp, micro-2024-ufc, micro-2024-trinity, isca-2025-fast, micro-2025-hawk, hpca-2026-unifhe, asplos-2026-falcon} provision hundreds of MBs of on-chip memory, enabling extensive reuse of \textevks under Min-KS.
Without Min-KS, BTS~\cite{isca-2022-bts} shows that off-chip memory accesses for \textevks dominate FHE runtime, effectively negating the benefits of hardware acceleration.

\noindent \textbf{$\bullet$ Capacity-aware BSGS fine-tuning~\cite{isca-2023-sharp}}. 
With $\#\mathrm{baby} = \mathcal{O}(\sqrt{\mathrm{\#diag}})$ baby ciphertexts, efficient BSGS requires the following working set ($\mathrm{WS}_{\mathrm{BSGS}}$) to fit within on-chip memory:
\begin{equation}
\label{eq:working-set}
\mathrm{WS}_{\mathrm{BSGS}} = (\mathrm{\#baby} + 1)\cdot \mathrm{Size}_{\mathrm{ciphertext}} + \mathrm{Size}_{\evk{}} + \mathrm{WS}_{\mathtt{KS}}.
\end{equation}
This capacity allows all the baby ciphertexts to remain on-chip for the entire giant step (GS in Eq.~\ref{eq:bsgs}).
GS multiplies the baby ciphertexts with streamed diagonal plaintexts to produce partial results, each immediately accumulated into the output ciphertext and followed by an $\mathtt{HRot}$ evaluation~\cite{micro-2022-ark}.
At the highest level, the ciphertext size ($\mathrm{Size}_{\mathrm{ciphertext}}$) reaches 24MiB and the \textevk size ($\mathrm{Size}_{\evk{}}$) reaches 60MiB.

To fit $\mathrm{WS}_{\mathrm{BSGS}}$ within its on-chip memory, SHARP~\cite{isca-2023-sharp} introduces BSGS fine-tuning, restricting $\mathrm{\#baby} \le 4$ for approximately 200MiB of on-chip memory.
This restricts the size of the working set: $\mathrm{WS}_{\mathrm{BSGS}} \le \mathrm{180MiB} + \mathrm{WS}_{\mathtt{KS}}$.
With memory-efficient scheduling (e.g., MAD~\cite{micro-2023-mad}) for $\mathtt{KS}$, $\mathrm{WS}_{\mathrm{BSGS}}$ can be held on-chip for the entire linear transform.

\noindent \textbf{$\bullet$ Modest $D_\mathrm{tr}$ for CtS matrix decomposition}. The choice of $D_\mathrm{tr}$ remained between 2 and 4 across various FHE algorithm~\cite{eurocrypt-2019-improved, cic-2025-overmodraise, asia-2025-subring, eurocrypt-2021-efficient}, application~\cite{wahc-2020-lattigo, icml-2022-resnet, asplos-2025-orion, github-liberate_FHE}, and hardware~\cite{hpca-2025-effact, hpca-2023-fab, hpca-2025-anaheim, isca-2023-sharp} studies.
Such $D_\mathrm{tr}$ choices allow moderately small $\mathrm{\#diag}$ values, such as 
15--31 for 4-level ($D_\mathrm{tr}=4$) CtS (baseline throughout this paper) shown in Table~\ref{tab:cts-analysis}.
This aligns well with BSGS fine-tuning because $\mathrm{\#baby}=4 \simeq \mathcal{O}(\sqrt{\mathrm{\#diag}})$.

\noindent \textbf{$\bullet$ Problems of conventional CtS}.
Despite optimizations such as BSGS fine-tuning, the large working set in Eq.~\ref{eq:working-set} remains a major constraint on performance and hardware efficiency.
Prior accelerators adopt decoupled data orchestration~\cite{isca-2022-craterlake} to enable software prefetching.
However, under tight on-chip memory budgets as in SHARP~\cite{isca-2023-sharp}, limited capacity leaves little room for prefetching, exposing off-chip memory access latencies that incur substantial overheads.

Later in the evaluation (Fig.~\ref{fig:sensitivity:util}), we quantitatively show that these memory bottlenecks severely limit hardware utilization.
In our baseline with over 200MiB of on-chip memory (\BASE), NTT units, which are the core compute engines of the architecture, achieve only 25.4\% utilization during \textboot, while off-chip memory (HBM) is utilized at 92.6\%.

Increasing on-chip memory capacity is a straightforward solution.
However, this approach increases area and energy costs and further entrenches memory-dominated designs.
Thus, \emph{rather than scaling memory capacity, we seek a principled solution that reduces the CtS working set and improves hardware efficiency at its root.}

\setlength{\tabcolsep}{4pt}
\begin{table}[t]
    \caption{Number of plaintexts (\#ptxt), \textevks (\#\textevk), and $\mathtt{KS}$ computations (\#$\mathtt{KS}$) needed as well as plaintext compression ratio (compr.) for CtS with Min-KS. \revb{Compressed size: theoretical\,/\,with our 32$\times$ cap (\S\ref{sec:arch:compress}); the capped value is used throughout.}}
    \label{tab:cts-analysis}
    \centering
    \begin{tabular}{l|ccc|cr@{\; }lc}
    \toprule
    \multicolumn{1}{c|}{\multirow{2}{*}{Level}}  &  \multicolumn{3}{c|}{Baseline: 4-level CtS} & \multicolumn{4}{c}{Optimized: 6(+1)-level fg-CtS}\\
    & $\mathrm{\#}\mathtt{KS}$ & $\mathrm{\#ptxt}$ & $\mathrm{\#}$\textevk & $\mathrm{\#}\mathtt{KS}$ & $\mathrm{\#ptxt}$ & (compr.) & $\mathrm{\#}$\textevk \\
    \midrule
       $\mathrm{intmd.}$ & - & - & - & 0 & 8 & (1$\times$) & 0 \\
       $\mathrm{top}$  & 6 & 16 & 2 & 6 & 7 & (8$\times$) & 1\\
       $\mathrm{top}-1$ & 6 & 15 & 2 & 6 & 7 & (32$\times$) &  1\\
       $\mathrm{top}-2$ & 10 & 31 & 2 & 6 & 7 &(128$\times$) & 1\\
       $\mathrm{top}-3$ & 10 & 31 & 2 & 6 & 7 &(512$\times$) &  1 \\
       $\mathrm{top}-4$ & - & - &- & 6 & 7 &(2048$\times$) &  1 \\
       $\mathrm{top}-5$ & - & - &- & 7 & 7 &(8192$\times$) &  1 \\
       \midrule
       \multicolumn{1}{c|}{Total ($\mathrm{\#}$)} & 32 & 93 & 8 & 37 & 50 & & 6\\
       \multicolumn{1}{c|}{(MiB)} & - & 1008 & 448 &  - &  \multicolumn{2}{c}{\ \ \ \ 451 $\rightarrow$ \textbf{23.5}\revb{\,/\,\textbf{31.5}}} & \textbf{324}\\
       \bottomrule
    \end{tabular}
    \vspace{-0.1in}
\end{table}
\setlength{\tabcolsep}{6pt}

\section{\NAME Cryptographic \& Algorithmic Enhancements}
\label{sec:alg}

We challenge conventional practices in cryptographic and algorithmic design of FHE, proposing novel algorithms and optimized parameter choices tailored to the memory capacity and bandwidth limitations in FHE accelerators.
Table~\ref{tab:cts-analysis} compares a conventional 4-level CtS configuration ($D_\mathrm{tr}=4$) with our optimized 6(+1)-level configuration ($D_\mathrm{tr}=7$). \revb{Table~\ref{tab:prior} summarizes the prior techniques we build upon or reuse.}

\subsection{Fine-grained CtS (fg-CtS)}
\label{sec:alg:fg-cts}

\emph{We depart from} the conventions in \emph{\S\ref{sec:motivation:final}; we set $D_{\mathrm{tr}}$ aggressively high (up to seven) and also eliminate the use of BSGS}, resorting to the basic linear transform method based on Eq.~\ref{eq:lt}.
We refer to this approach as fine-grained CtS (fg-CtS).
Through fine-grained decomposition of the CtS matrix, we limit each CtS level to \emph{at most eight non-zero diagonals}, making its computation sufficiently inexpensive even without BSGS.
At first glance, this appears disadvantageous; fg-CtS increases the total $\mathtt{KS}$ count and consumes two additional levels compared to the baseline 4-level CtS
as shown in Table~\ref{tab:cts-analysis}.

However, in systems constrained by memory capacity and bandwidth, the performance gains from eliminating BSGS outweigh the disadvantages.
Removing BSGS is equivalent to setting $\mathrm{\#baby}=1$ in Eq.~\ref{eq:working-set}, substantially reducing the working set per linear transform.
This frees on-chip memory, enabling aggressive software data prefetching for the next CtS level, which significantly decreases off-chip memory stalls in \textboot (shown later in Fig.~\ref{fig:nttu-stall-breakdown}).
Also, fg-CtS requires only one \textevk per CtS level and fewer diagonal plaintexts, reducing the total off-chip memory accesses to these data objects by 46.8\% ($(1008+448)\rightarrow(451+324)$ in Table~\ref{tab:cts-analysis}).

Table~\ref{tab:params} presents our optimized parameter selection for fg-CtS, which ensures 128-bit security~\cite{iacr-2024-guideline, jmc-2025-lattice-estimator}.
The additional level consumption of fg-CtS reduces \textefflevel, potentially increasing the frequency of \textboot.
However, state-of-the-art FHE ML studies~\cite{ccs-2024-neujeans, asplos-2025-orion, arxiv-2026-ckks-llama, icml-2024-he-transformer, access-2024-hyphen, asplos-2025-resbm} show a clear trend toward low \textefflevel ($\le10$, down to two~\cite{ccs-2024-neujeans}), with limited ability to exploit the extra levels of the baseline \revb{(e.g., Fig.~\ref{fig:bootcount})}.

\begin{table}[t]
\centering
\caption{\revb{Prior algorithmic techniques in \BASE and \NAME (new: fg-CtS, plaintext compression, intermediate ModRaise, KeyMult buffer, first EWE instruction).}}
\vspace{-0.07in}
\label{tab:prior}
\footnotesize
\setlength{\tabcolsep}{2pt}
\begin{tabularx}{\columnwidth}{@{}>{\raggedright\arraybackslash}p{0.475\columnwidth}>{\raggedright\arraybackslash}X@{\hspace{2pt}}l@{}}
\toprule
Prior technique & Role in \NAME & Use\\
\midrule
CtS/StC decomp.~\cite{access-2019-dft}, BSGS~\cite{crypto-2018-linear} & fg-CtS \textbf{(new)} & Built upon\\
Periodicity (WPC~\cite{ccs-2025-wpc}) & CtS/StC compression \textbf{(new)} & Built upon\\
OverModRaise~\cite{cic-2025-overmodraise} & Intmd.\ ModRaise \textbf{(new)} & Built upon\\
Sparse/subring secret~\cite{acns-2022-sparseboot, asia-2025-subring} & $\mathtt{KS}$-free first CtS level & Reused\\
Min-KS~\cite{micro-2022-ark} & One \textevk per fg-CtS level & Reused\\
Operation fusion~\cite{asplos-2026-cheddar} & Second EWE instr. & Reused\\
OF-Limb~\cite{micro-2022-ark} & Not used & --\\
\bottomrule
\end{tabularx}
%\vspace{-0.1in}
\end{table}
\setlength{\tabcolsep}{6pt}

\setlength{\tabcolsep}{4pt}
\begin{table}[t]
    \centering
    \caption{Baseline (Base) and optimized (Opt) FHE parameter sets. \textboot consumes 11--12 more levels in addition to CtS levels. \textefflevel is computed assuming scale $\Delta\simeq2^{40}$ after \textboot. All prime sub-moduli are chosen smaller than $2^{31}$, following Cheddar's methodology~\cite{asplos-2026-cheddar}. $\beta$ ($\mathtt{dnum}$ in \cite{rsa-2020-better}) is chosen to be four.}
    \label{tab:params}
    \begin{tabular}{c|cccccc}
        \toprule
        Param & $N$ & $Q_\mathrm{top}$ (\#$\mathcal{Q}_i$) & $P$ (\#$\mathcal{P}_i$) & $\beta$ & CtS levels & $L_\mathrm{eff}$\\
        \midrule
        Base & $2^{16}$ & $\sim2^{1373}$ (47) & $\sim2^{372}$ (12) & 4 & 4 & 12--13\\
        Opt & $2^{16}$ & $\sim2^{1374}$ (47) & $\sim2^{372}$ (12) & 4 & 6(+1) & {\color{white} 0}9--10\\
        \bottomrule
    \end{tabular}
    \vspace{-0.09in}
\end{table}
\setlength{\tabcolsep}{6pt}

\subsection{Plaintext compression}
\label{sec:alg:compress}

We further discover repetitive data patterns in CtS/StC plaintexts, allowing compression and decompression almost free of cost.
Unlike encrypted ciphertexts and \textevks, which must be indistinguishable from random values, plaintexts are unencrypted vectors encoded in a specific format and may exhibit data patterns amenable to compression.
Suppose a vector includes repetitive elements; \eg $\mathbf{u}=[X,Y,X,Y]\in\mathbb{C}^{N/2}$ for $N=8$.
When we encode $\mathbf{u}$ into a plaintext polynomial $u=\ptxt{\mathbf{u}}$ and apply NTT to it, we get a data pattern of $\mathrm{NTT}(u_{\mathcal{Q}_i}) = [A,B,C,D,A,B,C,D]$ for each limb.

This is formalized as Theorem~\ref{thm:compression}.
If the original vector contains $S$ distinct values repeated $\frac{N}{2S}$ times, the encoded polynomial has limbs with $2S$ values repeated $\frac{N}{2S}$ times, enabling $\frac{N}{2S}\times$ plaintext compression.
\begin{theorem}
\label{thm:compression}
Suppose $\mathbf{u}\in\mathbb{C}^{N/2}$ satisfies $\mathbf{u}[j]=\mathbf{u}[j+S]$ for all $0\le j < \frac{N}{2} - S$, where $S\mid\frac{N}{2}$.
Let $u=\ptxt{\mathbf{u}}\in\mathcal{R}_Q$ be its CKKS encoding, and let $\mathbf{t}_{\mathcal{Q}_i}=\mathrm{NTT}(u_{\mathcal{Q}_i})\in\mathbb{Z}_{\mathcal{Q}_i}^N$ be the NTT of each limb.
Then for all $0\le j < N - 2S$, $\mathbf{t}_{\mathcal{Q}_i}[j]=\mathbf{t}_{\mathcal{Q}_i}[j+2S]$.
\end{theorem}

We identify such repetitive patterns in CtS/StC matrix decomposition, enabling the compression of most CtS/StC plaintexts.
In the Fig.~\ref{fig:cts-optimizations} example, each diagonal of the rightmost decomposed matrix contains four values repeated four times ($S\!=\!4$ and $N\!=\!32$), allowing a 4$\times$ compression of its corresponding plaintext.
\emph{Our parameters enable up to an $8192\times$ compression of CtS plaintexts, reducing the total off-chip memory traffic for plaintexts to just 23.5MiB} \revb{(overall 19.2$\times$ compression) in theory, or 31.5MiB (14.3$\times$) under our 32$\times$ hardware cap.
Although overall compression is constrained by plaintexts lacking repetitions on the first CtS level (corresponding to the left decomposed matrix in Fig.~\ref{fig:cts-optimizations}), we alleviate this in \S\ref{sec:alg:intmd} (reflected in the results above).}

WPC~\cite{ccs-2025-wpc} also observes the same property as Theorem~\ref{thm:compression} and applies compression to ML weight plaintexts\revb{, which it must first repack into periodic patterns}.
\revb{We instead target} CtS/StC, where the memory impact is greater due to the high levels\revb{; there, the periodicity is intrinsic to the matrix decomposition, so each level's ratio is known in closed form (Table~\ref{tab:cts-analysis})}.
We also make accelerator-aware modifications.
As we store limbs in bit-reversed order for efficient automorphisms following ARK~\cite{micro-2022-ark}, the repetition pattern changes (e.g., $[A, A, C, C, B, B, D, D]$).
Nevertheless, we preserve the compressibility with appropriate hardware support (\S\ref{sec:arch:compress}).

\subsection{Intermediate ModRaise}
\label{sec:alg:intmd}
We propose intermediate ModRaise, a redesign of conventional ModRaise in \textboot that mitigates the limitations of fg-CtS and plaintext compression.
By executing the first CtS level at a separate intermediate modulus $Q_\mathrm{intmd}$ (e.g., $2^{140}$) instead of the top modulus $Q_\mathrm{top}$ (e.g., $2^{1400}$), we effectively reduce the level consumption of fg-CtS by one.
Moreover, the first-level CtS plaintexts use $Q_\mathrm{intmd}$; as these plaintexts are not compressible (see Table~\ref{tab:cts-analysis}), shrinking their modulus directly reduces their size and off-chip memory traffic.

Intermediate ModRaise requires several cryptographic adjustments.
Before \textboot, we switch to a ternary sparse secret~\cite{acns-2022-sparseboot} (Hamming weight: 48) from a subring of degree $N'=2^{13}$~\cite{asia-2025-subring}, which maintains 128-bit security based on our analysis with Lattice Estimator~\cite{jmc-2025-lattice-estimator}.
We extend the idea of OverModRaise~\cite{cic-2025-overmodraise} to realize intermediate ModRaise.
Instead of raising the modulus above $Q_\mathrm{top}$ (\eg to $\mathcal{P}_0\mathcal{P}_1Q_\mathrm{top}$ as in Fig.~\ref{fig:intmd-modraise}), we utilize a dedicated intermediate modulus $Q_\mathrm{intmd}$.
After performing the first CtS level under $Q_\mathrm{intmd}$, we raise the modulus again to $Q_\mathrm{top}$.
We briefly explain why this works:
\begin{theorem}
\label{thm:intmd-modraise}
Suppose the first CtS level combines $\mathrm{\#ptxt}$ plaintexts, each with coefficients bounded in absolute value by $B_\mathrm{ptxt}$.
A sufficient condition for intermediate ModRaise to produce the same result as executing the first CtS level directly over the top modulus is
\[
Q_\mathrm{intmd}\ge\mathrm{\#ptxt} \cdot B_\mathrm{ptxt} \cdot Q_\mathrm{bot} \cdot N.
\]
\end{theorem}

\begin{proof}[Proof outline]
Let $\ctxt{\mathbf{u}}=(a,b)\in\mathcal{R}_{Q_\mathrm{bot}}^2$.
Intermediate ModRaise \revb{simply reinterprets it with no arithmetic: $(a,b)\in\mathcal{R}_{Q_\mathrm{intmd}}^2$.}
The first CtS level computes
\[\textstyle
\left( \sum_{i=0}^{\mathrm{\#ptxt}-1} p_i \cdot a\big(\mathcal{X}^{5^{R_i}}\big), \sum_{i=0}^{\mathrm{\#ptxt}-1} p_i \cdot b\big(\mathcal{X}^{5^{R_i}}\big)\right),
\]
where $p_i$ are CtS plaintexts and $R_i$ are rotation amounts.
\revb{Using a subring secret of degree $N'$ ($\mathrm{\#ptxt}\le N/N'$), all $\mathtt{KS}$ computations are skipped for this level~\cite{asia-2025-subring}.}

\begin{figure}[t]
    \centering
    \includegraphics[width=0.98\columnwidth]{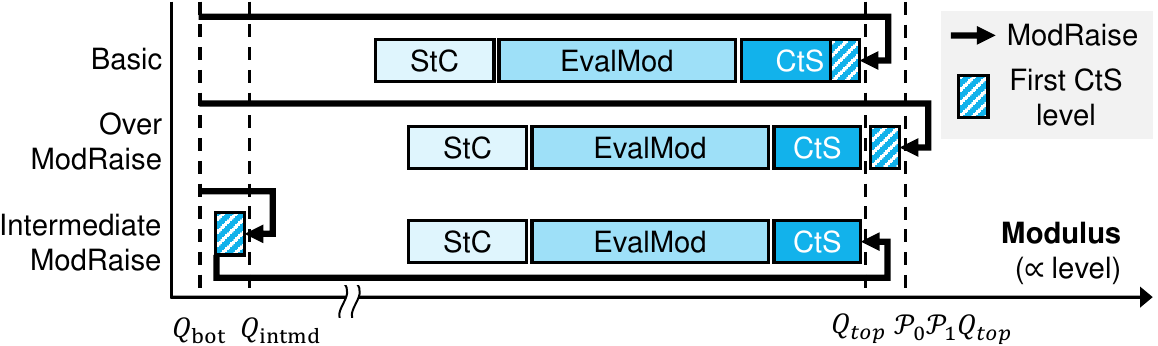}
    \vspace{-0.1in}
    \caption{\textboot modulus change across ModRaise methods.}
    \label{fig:intmd-modraise}
    \vspace{-0.05in}
    %\vspace{-0.1in}
\end{figure}

Correctness requires that this computation produces no modular wraparound under $Q_\mathrm{intmd}$.
Because coefficients of $a$ and $b$ are bounded by $Q_\mathrm{bot}/2$, automorphisms ($a\big(\mathcal{X}^{5^{R_i}}\big)$ and $b\big(\mathcal{X}^{5^{R_i}}\big)$) preserve this bound, and each plaintext coefficient is bounded by $B_\mathrm{ptxt}$, one plaintext-ciphertext multiplication contributes at most
$B_\mathrm{ptxt}\cdot(Q_\mathrm{bot}/2)\cdot N$
to each output coefficient.
Then, summing over $\mathrm{\#ptxt}$ plaintexts gives a bound of
$\mathrm{\#ptxt}\cdot B_\mathrm{ptxt} \cdot (Q_\mathrm{bot}/2) \cdot N$.
Choosing $Q_\mathrm{intmd}\ge\mathrm{\#ptxt} \cdot B_\mathrm{ptxt} \cdot Q_\mathrm{bot} \cdot N$ prevents the wraparound, making the computation over $\mathcal{R}_{Q_\mathrm{intmd}}$ identical to that over $\mathcal{R}_{Q_\mathrm{top}}$.
\end{proof}

With our parameter choices \revb{($\mathrm{\#ptxt}\!=\!8$)}, $Q_\mathrm{intmd}$ is small enough to be decomposed into only five sub-moduli primes; i.e., each plaintext for the first CtS level has only five limbs (1.25MiB).
\revb{Baseline first-level plaintexts at $Q_\mathrm{top}$ (188MiB with $\mathrm{\#ptxt}\!=\!16$) cap overall compression at 4.6$\times$. With only 1.25MiB per plaintext under $Q_\mathrm{intmd}$, any $D_\mathrm{tr}$ choice of 4--9 yields a 13.5--15.7$\times$ reduction in total CtS plaintext sizes under the 32$\times$ compression cap; i.e., intermediate ModRaise significantly augments compression regardless of fg-CtS.}

\section{Refining the Architecture with \NAME}
\label{sec:arch}

Starting from a baseline architecture resembling SHARP~\cite{isca-2023-sharp} (\BASE), we analyze its bottlenecks under \NAME's cryptographic and algorithmic techniques (\S\ref{sec:alg}) and propose architectural refinements to address them with minimal hardware overheads.
SHARP has likewise served as the baseline for more recent accelerator studies, such as Trinity~\cite{micro-2024-trinity}, FAST~\cite{isca-2025-fast}, and HAWK~\cite{micro-2025-hawk}.
%appliof  on accelerator architectures,

\subsection{Typical FHE accelerator architecture}
\label{sec:arch:sharp}

SHARP builds on prior designs: its vector cluster architecture and NTT unit (NTTU) follow F1~\cite{micro-2021-f1}, the BConv (BConvU) and automorphism (AutoU) units are adopted from ARK~\cite{micro-2022-ark}, and its decoupled data orchestration and scheduling methods derive from CraterLake~\cite{isca-2022-craterlake}.
Notably, F1's vector cluster architecture~\cite{micro-2021-f1} has since been widely adopted by numerous subsequent studies~\cite{isca-2022-craterlake, micro-2022-ark, isca-2023-sharp, isca-2025-fast, micro-2025-hawk, micro-2024-trinity}.

The vector cluster architecture deploys clusters with $\sqrt{N}$ vector lanes, enabling $\sqrt{N}$-way parallelism for its functional units (FUs).
The NTTU implements a four-step FFT datapath~\cite{sc-1989-fft4}, interpreting a length-$N$ limb as a $\sqrt{N}\times\sqrt{N}$ matrix.
It employs a long pipeline for the first row-wise $\sqrt{N}$-point NTTs and the subsequent column-wise NTTs, with inter-lane and inter-cluster data transfers to handle the required transposition in between.
The BConvU is composed of multiply-and-add (MAD) units arranged as a small output-stationary systolic array ($2\times6$ in our baseline) within each lane, well-suited for
the matrix-matrix multiplications in BConv~\cite{micro-2022-ark}.
SHARP introduces an element-wise engine (EWE) to support compound element-wise operations sharing input and output operands, reducing on-chip memory traffic.

The SHARP memory system includes a 180MiB main scratchpad shared by all FUs and an 18MiB BConv buffer dedicated to the NTTUs and BConvUs.
Two HBM stacks provide 1TB/s of off-chip memory bandwidth; 1TB/s has served as a reference point across numerous accelerator proposals.
HERACLES, a recently taped-out FHE accelerator from Intel, similarly employs two HBM3 stacks~\cite{isscc-2026-heracles}.

\subsection{Enhanced baseline architecture: \BASE}
\label{sec:arch:base}

We develop \BASE, an enhanced baseline reflecting recent FHE advances.
In particular, \BASE adopts a conventional 32-bit word size instead of SHARP’s 36-bit design, aligning with the rational-scaling approach of BitPacker~\cite{asplos-2024-bitpacker}.
It enhances hardware efficiency and enables flexible parameter selection.
We further adopt Cheddar’s 25-30 prime system~\cite{asplos-2026-cheddar}, an efficient BitPacker-based prime selection strategy.

Although SHARP originally includes four clusters, we scale \BASE to eight clusters to better balance computational throughput with memory bandwidth.
This configuration was also explored in the original SHARP paper and achieves a lower (better) energy-delay product (EDP).
To avoid fragmentation with more clusters, we adopt the overall data organization and parallelization strategy from CraterLake~\cite{isca-2022-craterlake}, which ensures an equal distribution of jobs across the clusters and the lanes.
As in CraterLake, we employ fixed-wire on-chip networks (NoCs) for this data organization.

\begin{figure}[t]
    \centering
    \includegraphics[width=0.967\columnwidth]{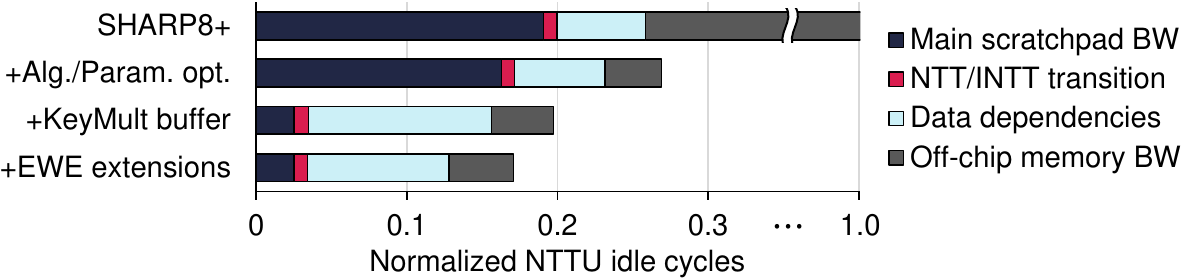}
    \vspace{-0.07in}
    \caption{Breakdown of the NTTU idle cycles during \textboot. We start from \BASE using the baseline parameters in Table~\ref{tab:params}.
    We %gradually
    add our algorithmic and parameter optimizations (\S\ref{sec:alg}), KeyMult buffer, and EWE extensions.}
    \label{fig:nttu-stall-breakdown}
    \vspace{-0.02in}
\end{figure}

\subsection{Bottleneck analysis}
\label{sec:arch:bottleneck}

F1-based architectures aim at maximizing NTTU utilization, as NTT is the dominant polynomial operation in FHE and entails both heavy computation and costly chip-wide data movement.
By adopting a very long instruction word (VLIW) architecture, they overlap other polynomial operations with NTT execution to hide their latency.

Nevertheless, NTTUs still exhibit frequent idle periods.
Using our simulation framework (\S\ref{sec:eval:setup}), we analyze the sources of NTTU idle cycles in \BASE.
Fig.~\ref{fig:nttu-stall-breakdown} shows the analysis results.
We identify four primary causes:
\begin{itemize}[leftmargin=*]
\item Off-chip memory bandwidth (BW) stall: NTTU waits for data from the off-chip memory.
\item Data dependency stall: NTTU waits for FU outputs.
\item Main scratchpad BW stall: All the ports in the main scratchpad are occupied by the EWE or the AutoU.
\item NTT/INTT transition overhead: Switching between NTT and inverse NTT (INTT) introduces pipeline bubbles.
\end{itemize}

For \BASE, off-chip memory bandwidth stalls dominate NTTU idle cycles, 70\% of which are eliminated by our techniques in \S\ref{sec:alg}.
Once this bottleneck is resolved, data dependencies and the limited bandwidth of on-chip memory (main scratchpad) emerge as new bottlenecks.

In the following sections, we propose two architectural refinements to address these issues.

\begin{figure}[t]
    \centering
    \includegraphics[width=0.99\columnwidth]{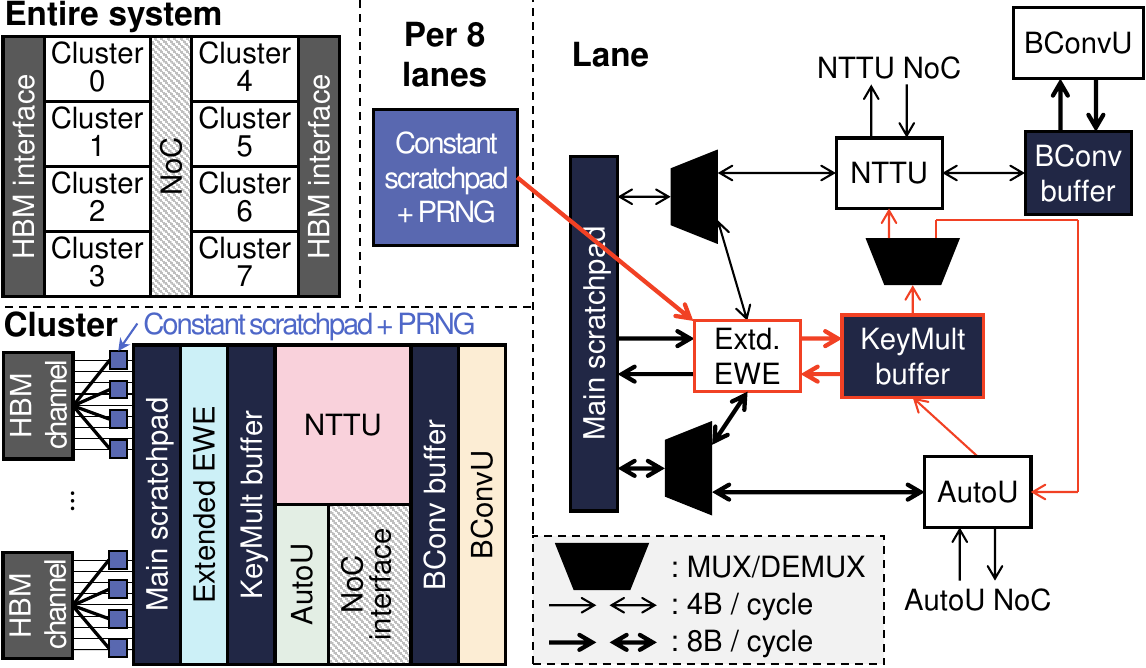}
    \caption{Applying \NAME to an accelerator with eight clusters, each with $\sqrt{N}=256$ vector lanes and functional units for NTT (NTTU), BConv (BConvU), automorphism (AutoU), and element-wise operations (EWE). \NAME's architectural refinements are highlighted in red.}
    \label{fig:arch}
\end{figure}

\subsection{KeyMult buffer \& reduced on-chip memory}
\label{sec:arch:memory}

While prior work has focused on NTT and BConv,
element-wise operations dominate on-chip memory bandwidth usage.
Specifically, multiplying a polynomial ($d$) with an \textevk, referred to as KeyMult, is a frequent, memory-intensive step in every $\mathtt{KS}$.
KeyMult is defined as:
\begin{equation}
\label{eq:keymult}
\text{KeyMult:\ } (a,b)\gets\Big(\sum_{i=0}^{\beta-1}d\cdot\evk{}[0][i], \sum_{i=0}^{\beta - 1} d\cdot\evk{}[1][i]\Big)
\end{equation}
An \textevk is composed of $2\times\beta$ (\eg $\beta=4$~\cite{rsa-2020-better}) polynomials.
Half of them ($\evk{}[0][i]$) are deterministic random polynomials that can be generated by a hardware pseudo-random number generator (PRNG)~\cite{isca-2022-craterlake}.
Due to PRNG throughput and scratchpad bandwidth limits, KeyMult is executed iteratively over $\beta$ steps in the EWE, where each step computes $(a,b)\gets(a+d\cdot\evk{}[0][i], b + d\cdot\evk{}[1][i])$.

As the $(a,b)$ data is frequently accessed for output accumulation, we introduce the KeyMult buffer, a dedicated on-chip memory for the $(a,b)$ data.
This separate hardware buffer provides additional input/output ports to the FUs (see Fig.~\ref{fig:arch}), boosting the aggregate on-chip memory bandwidth.

To offset the KeyMult buffer's hardware overhead, we shrink the main scratchpad from SHARP's 180MiB~\cite{isca-2023-sharp} to 128MiB based on design space exploration; with fg-CtS's smaller working set, 128MiB plus the 32MiB KeyMult buffer still holds the CtS working set (Eq.~\ref{eq:working-set} with $\mathrm{\#baby}=1$) and leaves room for aggressive software prefetching.

\revb{The KeyMult buffer is sliced by the lanes (default: 16KiB per lane, 4MiB per cluster) without cross-cluster sharing.
Each 16KiB per-lane slice is composed of eight single-ported SRAM banks, serving six 32-bit ports through bank interleaving~\cite{tc-1985-interleave}: 2R2W (EWE) + 1R1W (NTTU and AutoU). It provides 48TB/s of chip-wide bandwidth at 1GHz, which is less than the main scratchpad (eight ports, 64TB/s).}

The KeyMult buffer lowers the average energy per memory access and eliminates most of the main scratchpad bandwidth stalls (Fig.~\ref{fig:nttu-stall-breakdown}).
However, data dependency stalls, previously masked by bandwidth stalls, become more pronounced.

\subsection{Element-wise engine (EWE) extensions}
\label{sec:arch:ewe}

Most data dependency stalls originate from waiting on the EWEs.
The EWE design of SHARP deploys four modular MAD (MMAD) units per lane, which can collaborate to execute various compound element-wise instructions.
However, the EWEs are seldom fully utilized due to the limited bandwidth of the main scratchpad.
For instance, executing four independent MMADs would require loading 12 operands and storing 4 results per lane per cycle, whereas the main scratchpad can supply only 7 (see Fig.~\ref{fig:arch}).
Thus, discovering compound instructions that share operands is crucial for maximizing the effective throughput of the EWEs.

We leverage the KeyMult buffer's additional bandwidth to propose the following instruction extensions for the EWE:
\begin{enumerate}[leftmargin=*]
\item $d_\mathrm{res}\gets(p\cdot a + a')$ \&\newline $(a_\mathrm{res}, b_\mathrm{res})\gets (d_\mathrm{res} \cdot \evk{}[0][i], d_\mathrm{res} \cdot \evk{}[1][i] + (p \cdot b + b'))$
\item $a_\mathrm{res}\gets(C\cdot a - a')\cdot C'$ for constants $C$ and $C'$
\end{enumerate}
The first instruction combines $\mathtt{PMult}$ and $\mathtt{HAdd}$ with KeyMult by leveraging all four MMAD units in an EWE.
It is heavily used for evaluating Eq.~\ref{eq:lt} for fg-CtS or other linear transform evaluations.
Here, the KeyMult buffer handles the input $b'$ and the output $(a_\mathrm{res}, b_\mathrm{res})$.
The second instruction implements an operation fusion technique in Cheddar~\cite{asplos-2026-cheddar}.

While these two extensions eliminate 23\% of the data dependency stalls for the NTTUs (see Fig.~\ref{fig:nttu-stall-breakdown}),
our focus is on demonstrating potential architectural efficiency improvements from such instruction extensions, which only incur minimal control logic complications.
As shown in Table~\ref{tab:area}, \NAME's two EWE extensions add only 0.4mm\textsuperscript{2} to a total chip area of $\sim$190mm\textsuperscript{2}.
Future work could explore additional EWE extensions based on common computational patterns.

\subsection{Supporting the \NAME algorithms}
\label{sec:arch:compress}

We provide architectural support to maximize the benefits of plaintext compression and intermediate ModRaise.
First, a small (6MiB) scratchpad, implicitly used for constants in SHARP, is also employed to store compressed plaintexts.
Each group of eight lanes shares a constant scratchpad SRAM that broadcasts to all eight lanes, and we use this datapath for compressed plaintexts.
\revb{In bit-reversed order, repeated values are adjacent (\S\ref{sec:alg:compress}); thus, the broadcast itself expands the compressed plaintexts without extra decompression logic}.
Four groups of SRAM are connected to an HBM channel (or pseudo-channel), making each channel handle 32 lanes.
This design supports compression rates from 8$\times$ to 32$\times$ with minimal hardware overhead while enabling plaintexts to be loaded at full HBM bandwidth via all channels.
Higher rates would bring diminishing returns.

Further, intermediate ModRaise requires performing an RNS reconstruction with an input polynomial having 2--3 limbs, which recovers large coefficients from their residues modulo $\mathcal{Q}_i$.
To enable this, we extend the BConvU to incorporate the necessary reconstruction logic.

\section{Evaluation}
\label{sec:eval}

\setlength{\tabcolsep}{4pt}
\begin{table}[t]
    \centering
    \caption{Hardware configuration and chip area breakdown of \BASE (SH8+) and \NAME.}
    \label{tab:area}
    \begin{tabularx}{0.99\columnwidth}{c|LL}
        \toprule
         & \tablecenter{\textbf{\BASE}} & \tablecenter{\textbf{\NAME}}\\
        \midrule
        Overall & \multicolumn{2}{>{\hsize=\dimexpr2\hsize+2\tabcolsep+\arrayrulewidth\relax}X}{8 clusters with a total of 2048 vector lanes:
        \newline+ 8 NTTUs (31.25M NTT/s) w/ dedicated NoC
        \newline+ BConvU: $2\times6$ plain MAD units per lane (24576 GOPS)
        \newline+ 8 AutoUs (31.25M automorphism/s) w/ dedicated NoC
        \newline+ EWE: 4 MMAD units per lane (8192 GOPS)}\\
        \midrule
        SRAM & 180MiB (Main, 64TB/s) \newline+ 18MiB (BConv, 40TB/s) \newline+ 6MiB (Constant) & {128MiB (Main, 64TB/s) \newline+ 32MiB (KeyMult, 48TB/s) \newline+ 18MiB (BConv, 40TB/s) \newline+ 6MiB (Constant)}\\
        \midrule
        HBM & 1024GB/s (2 HBM3 stacks) & 1024GB/s (2 HBM3 stacks)
    \end{tabularx}
    \begin{tabularx}{0.99\columnwidth}{@{\hspace{2pt}}l@{\hspace{2pt}}|@{\hspace{2pt}}C@{}C@{}C@{}C@{}C@{}C@{}C@{}C@{\hspace{2pt}}|@{\hspace{2pt}}c@{\hspace{2pt}}}
\midrule
    \multicolumn{10}{c}{\textbf{Area (mm\textsuperscript{2})}}\\
\midrule
    & NTT & BConv & Auto & EWE & PRNG & NoC & SRAM & HBM & Total\\
\midrule
    SH8+ & 31.4 & 23.8 & 6.8 & 6.3 & 2.3 & 16.6 & 76.2 & 29.6 & 193.1\\
    \NAME & 31.4 & 23.8 & 6.8 & 6.7 & 2.3 & 16.6 & 73.0 & 29.6 & 190.2\\
\bottomrule
    \end{tabularx}
\end{table}
\setlength{\tabcolsep}{6pt}

\subsection{Experimental setup and area}
\label{sec:eval:setup}

We developed a simulator (and scheduler) for our architecture with a minimum granularity of eight cycles.
It maps a given instruction list to the FUs while avoiding structural and data hazards, adopting the decoupled data orchestration~\cite{asplos-2019-decoupled} of CraterLake~\cite{isca-2022-craterlake}.
Once all instruction timings are fixed, it schedules data movement between the off-chip and on-chip memory with Belady's MIN policy~\cite{ibm-1966-belady}. When software prefetching cannot bring data in time, the simulator inserts static deterministic stalls.
\revb{As FHE has no data-dependent control flow, all memory accesses are known at compile time, letting prior FHE accelerators schedule data movement statically~\cite{micro-2021-f1, isca-2022-bts, isca-2022-craterlake, micro-2022-ark, isca-2023-sharp, seed-2024-cifher, hpca-2025-effact, asplos-2026-falcon, asplos-2025-cinnamon}; F1, CraterLake, SHARP, and Cinnamon also use Belady's MIN.
Results are insensitive to the replacement policy: with LRU, \textboot time changes by less than 1\% for both \BASE and \NAME at a matched $\efflevel$.}
Our simulator collects statistics of the system for the given list of instructions and calculates utilization and energy consumption of each resource.

For energy and area estimation, we synthesized the core logic units in RTL using the ASAP7 7.5-track 7nm predictive process design kit (PDK)~\cite{mj-2016-asap7}.
Most hardware components operate at 1GHz, but the NoCs and a few small buffers run at 2GHz.
We estimated the energy for mm-scale global wires based on Dally et al.~\cite{vlsic-2018-wire}.
SRAM components were modeled using FinCACTI~\cite{isvlsi-2014-fincacti}, which we extensively modified to reflect published data for 7nm technologies~\cite{isscc-2017-7nm-sram, iedm-2017-gf7nm, vlsit-2018-samsung, isscc-2018-7nm-sram-euv, iedm-2016-tsmc7nm, whitepaper-2018-irds, isca-2021-tpuv4i}.
All SRAM components are single-ported; we used bank interleaving to achieve high bandwidth at the cost of coarser access granularity.
Two HBM3 stacks, each with 512GB/s of bandwidth~\cite{jedec-2022-hbm3}, were employed, whose energy and PHY area were estimated using prior studies~\cite{micro-2017-finegrainedDRAM, isca-2021-tpuv4i}.

Table~\ref{tab:area} summarizes the hardware configuration and area breakdown of \NAME and \BASE.
Owing to minimal architectural overheads and reduced on-chip memory, \NAME reduces the total chip area from 193.1mm\textsuperscript{2} to 190.2mm\textsuperscript{2}.

\subsection{Workloads and parameters}
\label{sec:eval:workload}

We used the following four workloads, which have been widely used in prior work:
\begin{itemize} %[leftmargin=*]
    \item \textbf{\textboot} denotes bootstrapping refreshing $2^{15}$ complex numbers.
    \item \textbf{HELR}~\cite{aaai-2019-helr} performs  binary classification model training using logistic regression on $14\!\times\!14$ grayscale images for 32 iterations, each with a batch size of 1,024.  
    \item \textbf{Sorting}~\cite{tifs-2021-sorting} is an implementation of a two-way bitonic sorting network for $2^{14}$ real numbers.
    \item \textbf{ResNet-20 (RN-20)} is an implementation of CNN inference by Lee et al.~\cite{icml-2022-resnet} on a $32\times 32 \times 3$ CIFAR-10~\cite{techreport-2009-cifar10} image using the ResNet-20~\cite{cvpr-2016-resnet} model.
\end{itemize}

We used the baseline/optimized parameters listed in Table~\ref{tab:params} for \BASE/\NAME.
The respective \textefflevel values for \textboot are 12 and 9.
For the other workloads, \textefflevel is adjusted as appropriate (8--10) based on each workload's requirement\revb{, identically for \BASE and \NAME (ResNet-20: 8; HELR and Sorting: 10; Orion CNNs in \S\ref{sec:eval:cnn}: 10)}.
For HELR, a much smaller $2^8\times2^8$ CtS matrix can be used instead of the full $\frac{N}{2}\times\frac{N}{2}$ matrix.
Thus, we used fewer matrices to decompose it, resulting in a 4-level CtS even for fg-CtS.

\noindent\revb{\textbf{Precision.} Our techniques change only \textboot and compute the same function (Theorem~\ref{thm:intmd-modraise}): with the Cheddar GPU library~\cite{asplos-2026-cheddar} ($\Delta=2^{40}$, 100 runs), the $\log_2$ maximum error of \textboot is $-19.86$ without and $-19.84$ to $-19.85$ with them (including the use of sparse subring secret).}

\setlength{\tabcolsep}{4pt}
\begin{table}[t]
    \centering
    \caption{Performance comparison with prior ASIC FHE accelerators. We also calculate area-delay product (ADP), calculate its relative value compared to that of \NAME, and show its geometric mean across the workloads (vs. \NAME). }
    \label{tab:eval-workloads}
    \begin{tabularx}{0.99\columnwidth}{l|RRRR|rr}
       \toprule
       \multirow{2}{*}{Hardware} & \tablecenter{\textboot}  & \tablecenter{HELR} & \tablecenter{Sorting} & \multicolumn{1}{c|}{RN-20} & \tablecenter{Area} & \tablecenter{vs. \NAME} \\
        & \tablecenter{(ms)} & \tablecenter{(ms/it)} & \tablecenter{(s)} & \multicolumn{1}{c|}{(ms)} & \tablecenter{(mm\textsuperscript{2})} & \tablecenter{(ADP)} \\
        \midrule
        GPU\textsuperscript{$\dagger$} & 14.0{\color{white}00} & 18.0{\color{white}0} & \tablecenter{*-} & 472{\color{white}.0} & \tablecenter{-} & \tablecenter{-}\\        
       \midrule
         CrLake~\cite{isca-2022-craterlake} & 6.33{\color{white}0} & *3.81 & \tablecenter{*-} & *321{\color{white}.0} & 223 & 8.08$\times$\\
         ARK~\cite{micro-2022-ark} & 3.56{\color{white}0} & 7.42 & 1.99 & 125{\color{white}.0} & 418 & 8.74$\times$\\
         SHARP~\cite{isca-2023-sharp} & 3.12{\color{white}0} & 2.53 & 1.38 & 99.0 & 179 & 2.38$\times$\\
         UFC~\cite{micro-2024-ufc} & 2.64{\color{white}0} & 2.10 & 1.16 & *87.6 & 198 & 2.21$\times$ \\
         Trinity~\cite{micro-2024-trinity} & 1.92{\color{white}0} &  1.37 & \tablecenter{*-} & 89{\color{white}.0} & \textsuperscript{$\ddagger$}157 & 1.38$\times$ \\
         FAST~\cite{isca-2025-fast}  & 1.38{\color{white}0} & 1.33 & \tablecenter{*-} & 60.5 & 284 & 1.95$\times$\\
         HAWK~\cite{micro-2025-hawk}  & 2.15{\color{white}0} & 1.96 & 0.98 & 74.9 & 186 & 1.81$\times$\\
         UniFHE~\cite{hpca-2026-unifhe} & 2.8{\color{white}00} & 2.23 & \tablecenter{*-} & 88.7 & 149 & 1.74$\times$\\
         Falcon~\cite{asplos-2026-falcon} & 2.11{\color{white}0} & 1.94 & 0.96 & 72.8 & 180 & 1.72$\times$\\
         \midrule
         \BASE &  2.50{\color{white}0} & 1.85 & 1.12 & 75.6 &  193  & 2.01$\times$\\
         \NAME &  0.915 & 1.40 & 0.52 & 38.7 & 190 & 1$\times$ \\
       \bottomrule
    \end{tabularx}
    \begin{itemize}[leftmargin=*]
    \item[*] We excluded these from the ``vs. \NAME'' geometric mean calculation due to substantial differences in workload implementations or missing data.
    \item[$\dagger$] Based on Theodosian~\cite{ispass-2026-theodosian} on B200.
    \item[$\ddagger$] Trinity's 36-bit NTTUs occupy 39.5\% of the area of our 32-bit NTTUs at the same throughput, and its 180MiB main scratchpad area is 64.1\% of that of \BASE, implying substantial differences in evaluation settings.
    \end{itemize}
\end{table}
\setlength{\tabcolsep}{6pt}

\begin{figure*}
    \centering
    \subfloat[Time \& energy (\textboot)]{\includegraphics[width=0.435\columnwidth]{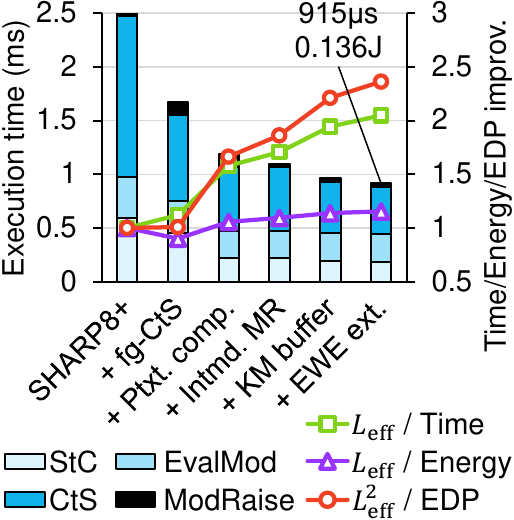}\label{fig:sensitivity:boot}}
    \quad
    \subfloat[Time \& energy (workloads)]{\includegraphics[width=0.955\columnwidth]{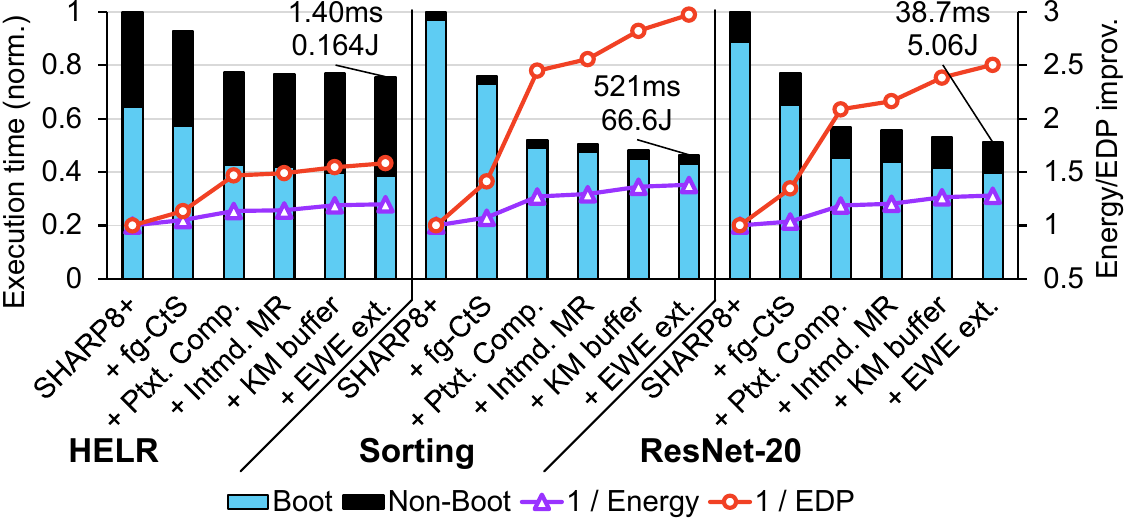}\label{fig:sensitivity:workload}}
    \quad
    \subfloat[Utilization (\textboot)]{\includegraphics[width=0.54\columnwidth]{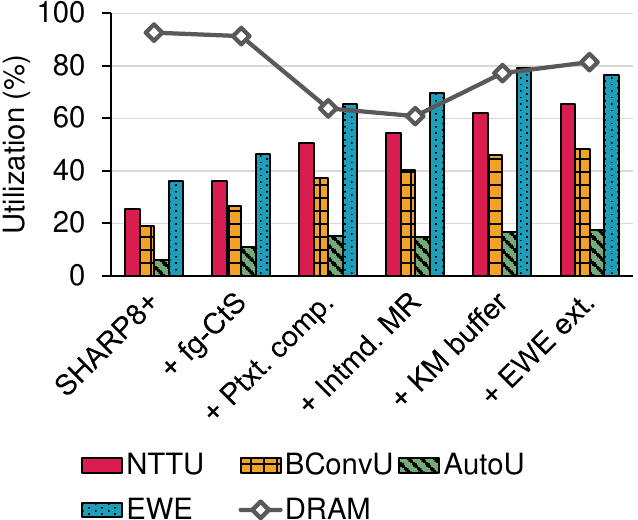}\label{fig:sensitivity:util}}
    \caption{Execution time, energy consumption, and energy-delay product (EDP) of (a) \textboot and (b) the other three workloads when we incrementally apply fg-CtS (+ fg-CtS), plaintext compression (+ Ptxt. Comp.), intermediate ModRaise (+ Intmd. MR), KeyMult buffer (+ KM buffer), and EWE extensions (+ EWE ext.). \textboot time and energy are divided by \textefflevel for a fair comparison. (c) Utilization of the hardware resources depending on the configuration during \textboot is also shown.
    To reflect the fact that lower is better for energy and EDP, we plot their reciprocals.
    }
    \label{fig:sensitivity}
    \vspace{-0.07in}
\end{figure*}

\subsection{\NAME vs. prior ASIC accelerators}
\label{sec:eval:vs-prior}

Table~\ref{tab:eval-workloads} demonstrates that \NAME achieves significant performance improvements over state-of-the-art FHE accelerators~\cite{isca-2022-craterlake, micro-2022-ark,isca-2023-sharp,micro-2024-trinity,micro-2024-ufc,micro-2025-hawk,isca-2025-fast, hpca-2026-unifhe, asplos-2026-falcon}.
\NAME achieves 12.2--15.3$\times$ speedups over Theodosian~\cite{ispass-2026-theodosian} on B200.

Notably, we achieve the first-ever \textbf{sub-millisecond CKKS bootstrapping} ($\efflevel\!=\!9$), which corresponds to a 1.51$\times$ acceleration over the previous state-of-the-art (FAST~\cite{isca-2025-fast})\revb{; prior work reports \textboot at $\efflevel=8$ (e.g.,~\cite{micro-2022-ark,isca-2023-sharp,isca-2025-fast,micro-2025-hawk}).}
\revb{Also, compared to \BASE, \NAME improves per-\textefflevel \textboot time by 2.05$\times$ ($\frac{\text{0.915ms}}{\text{9}}$ vs. $\frac{\text{2.50ms}}{\text{12}}$), or 2.36$\times$ (0.915ms vs.\ 2.16ms) when using the same $\efflevel=9$ for both}.
Compared to prior ASIC accelerators, \NAME achieves 1.84--3.82$\times$ and 1.56--3.23$\times$ speedups for sorting and ResNet-20, respectively.
As our optimizations target \textboot, \NAME shows smaller improvements on HELR, which contains a lower \textboot portion (see Fig.~\ref{fig:sensitivity:workload}) and utilizes a reduced CtS matrix size.

When also accounting for area, \NAME shows 1.38--8.74$\times$ reductions in area-delay product (ADP), measured as the geometric mean across the workloads.
While on-chip memory accounts for the largest area portion (38.4\%) in \NAME, its total capacity (184MiB, see Table~\ref{tab:area}) is smaller than that of recent accelerators, such as UFC (274MiB), FAST (281MiB) and HAWK (212MiB), allowing \NAME to achieve a compact area while maintaining high performance.

\begin{figure}
    \centering
    \includegraphics[width=0.97\columnwidth]{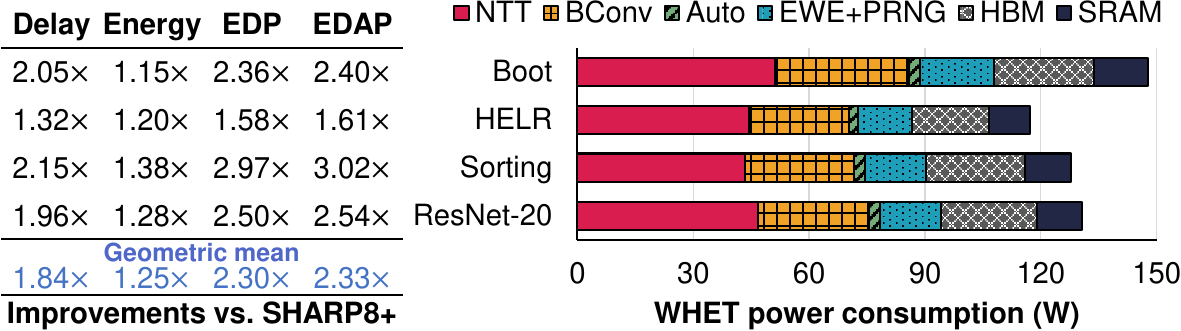}
    \vspace{-0.07in}
    \caption{(Left) Delay, energy, energy-delay product (EDP), and energy-delay-area product (EDAP) improvements from applying \NAME on \BASE. (Right) Average power consumption breakdown of \NAME for the workloads. NTT and automorphism include NoC power. \revb{\textboot delay and energy are divided by \textefflevel.}}
    \label{fig:power}
    \vspace{-0.07in}
\end{figure}

\subsection{Energy and efficiency of \NAME (vs. \BASE)}
\label{sec:eval:power}

Fig.~\ref{fig:power} (right) shows the power consumption breakdown of \NAME across the workloads.
\NAME consumes 117.1--147.7W of power in total.
While differences in evaluation settings make direct comparison difficult, the power consumption of \NAME is within a comparable range to prior ASIC proposals, including SHARP (97W for \textboot), Trinity (229W peak), and FAST (120--160W for the workloads).
NTT is the largest contributor, accounting for 33.9--38.0\% of the total power.
This is consistent with Fig.~\ref{fig:level:compute}, where NTT accounts for over half of the computational complexity in $\mathtt{KS}$.
Next, the memory system (HBM + SRAM) accounts for 26.0--27.9\%, highlighting the importance of our memory optimizations.

As shown in Fig.~\ref{fig:power} (left), we compare \NAME with a controlled baseline, \BASE, which ensures a fair evaluation of hardware efficiency.
\NAME achieves geometric mean reductions of 1.25$\times$ in energy, 2.30$\times$ in energy-delay product (EDP), and 2.33$\times$ in energy-delay-area product (EDAP~\cite{micro-2009-mcpat}).
%reduction in terms of the geometric mean.
Combined with 1.32--2.15$\times$ speedups, this demonstrates that \NAME enables efficient execution of various FHE workloads.

\subsection{Ablation study: from \BASE to \NAME}
\label{sec:eval:vs-base}

Fig.~\ref{fig:sensitivity} shows how performance, energy, and hardware utilization change from \BASE to \NAME as we incrementally apply each of our optimizations (\S\ref{sec:alg} and \S\ref{sec:arch}).

Our cryptographic and algorithmic optimizations result in significant improvements for \BASE with combined benefits of 1.30--\revb{1.98}$\times$ speedups and 1.09--1.29$\times$ energy reductions.
Despite the reduced \textefflevel (12 to 9), fg-CtS accelerates ModRaise+CtS by 1.65$\times$ and per-\textefflevel \textboot by 1.12$\times$, which translates to 1.08--1.32$\times$ speedups for the other three workloads.
Plaintext compression reduces both time and energy, improving EDP by 1.29--1.73$\times$.
Intermediate ModRaise\revb{, which replaces OverModRaise~\cite{cic-2025-overmodraise} used in the two preceding steps, speeds up ModRaise and the first CtS level by 3.81$\times$}, yielding an overall 8.2\% \textboot performance improvement.

The architectural refinements result in additional EDP improvements of 1.06--1.27$\times$.
The KeyMult buffer and the EWE extensions achieve 4.2--5.3\% and 0.8--1.4\% energy savings, respectively, and together 1.01--1.20$\times$ speedups.

Fig.~\ref{fig:sensitivity:util} shows that our optimizations deliver consistent improvements in FU utilization.
The off-chip memory bandwidth bottleneck keeps every FU of \BASE below 36.2\%, whereas \NAME reaches 65.3\%, 48.4\%, and 76.5\% for the NTTU, BConvU, and EWE with 81.3\% HBM utilization, a balanced use of computation and memory bandwidth.

\begin{figure}
    \centering
    \subfloat[\revb{Memory system}]{\includegraphics[width=0.38\columnwidth]{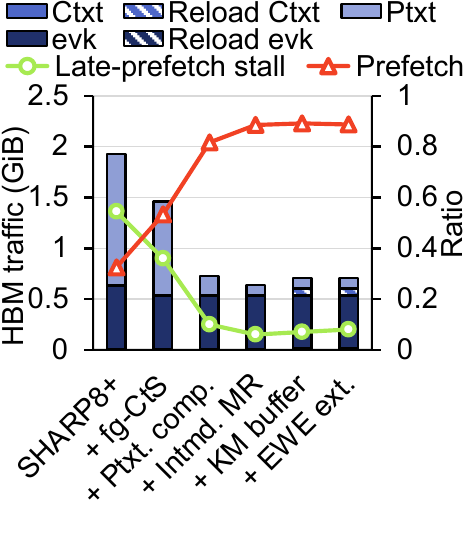}\label{fig:breakdown:memory}}
    \subfloat[\revb{Energy consumption}]{\includegraphics[width=0.33\columnwidth]{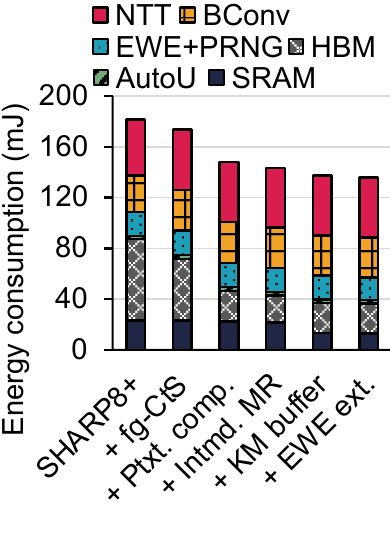}\label{fig:breakdown:energy}}
    \subfloat[\revb{Standalone}]{\includegraphics[width=0.23\columnwidth]{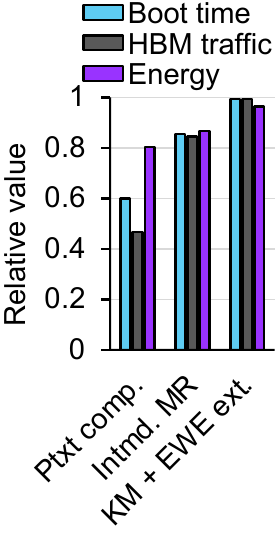}\label{fig:breakdown:standalone}}
    \caption{\revb{(a) Memory system and (b) energy breakdowns for \textboot with incremental optimizations (as in Fig.~\ref{fig:sensitivity}) at a matched $\efflevel\!=\!9$. HBM traffic comprises Ctxt (ciphertext), Ptxt (plaintext), \textevk, and reloads. Late-prefetch stall is the fraction of \textboot time waiting for late reads; Prefetch is the fraction of HBM read bytes available before use (prefetch coverage). (c) Each optimization alone for \textboot on \BASE without fg-CtS ($\efflevel\!=\!12$), vs.\ \BASE}.}
    \label{fig:breakdown}
    \vspace{-0.07in}
\end{figure}

\noindent\revb{\textbf{Memory system.}\label{sec:eval:mem} Fig.~\ref{fig:breakdown:memory} and Fig.~\ref{fig:breakdown:energy} break down HBM traffic by data type and energy for each optimization step.
In \BASE, plaintexts account for 67\% of the HBM traffic (vs. 15\% in \NAME).
Also, as its CtS working set (176MiB) nearly fills the 180MiB scratchpad, \BASE
prefetches only 32\% of its HBM read traffic in time.
fg-CtS removes BSGS from CtS and shrinks the working set to 121MiB, raising the prefetch coverage to 53\%.
Compression then halves the HBM traffic (1.46$\rightarrow$0.73GiB), raising the coverage to 82\%.
Reloads appear only after the main scratchpad shrinks to 128MiB for the KeyMult buffer (ciphertexts 36MiB, \textevks 4.5MiB).
Overall, HBM traffic falls from 1.93 to 0.71GiB, late-prefetch stalls from 54 to 8\% of \textboot time, and HBM energy from 64.4 to 23.8mJ.}
\revb{Applied standalone to \BASE (Fig.~\ref{fig:breakdown:standalone}), plaintext compression and intermediate ModRaise yield 1.67/1.17$\times$ speedups, cutting HBM traffic to 1.04/1.87GiB (from 2.21GiB) and energy by 20/13\%. Conversely, the KeyMult buffer and EWE extensions (the latter needs the former, \S\ref{sec:arch:ewe}) require prior HBM traffic reduction to be effective (1.01$\times$ speedup alone vs. 1.10$\times$ after compression).}

\begin{figure}
    \centering
    \subfloat[\revb{$D_\mathrm{tr}$ sweep}]{\includegraphics[width=0.48\columnwidth]{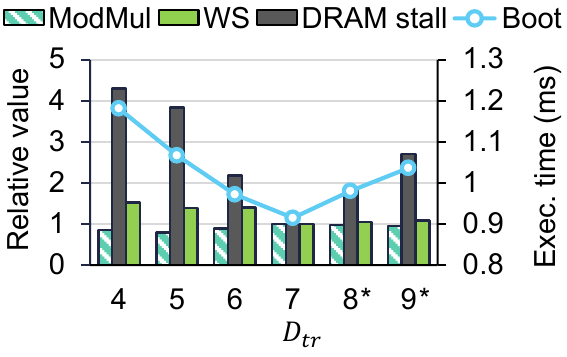}\label{fig:param-sweep:Dtr}}
    \subfloat[\revb{KeyMult buffer capacity}]{\includegraphics[width=0.46 \columnwidth]{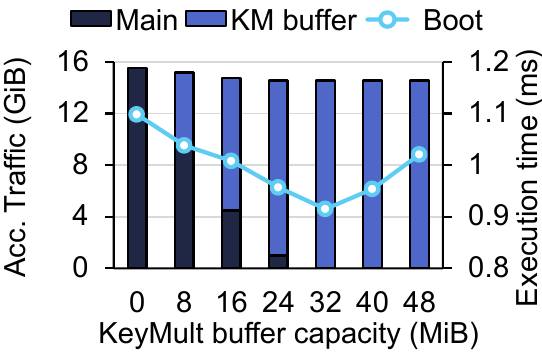}\label{fig:param-sweep:KM}}
\caption{\revb{(a) \(D_{\mathrm{tr}}\) sweep on WHET (\(L_{\mathrm{eff}}=9\)); modular multiplication count (ModMul), CtS working set (WS), and DRAM stalls are normalized to \(D_{\mathrm{tr}}=7\). An asterisk (*) indicates configurations that exceed the modulus limit for 128-bit security~\cite{iacr-2024-guideline}. (b) Main scratchpad (Main) plus KeyMult buffer (KM buffer) capacity is fixed at 160MiB.}
    }
    \label{fig:param-sweep}
    \vspace{-0.06in}
\end{figure}

\subsection{\revb{Choice of $D_\mathrm{tr}$}}
\label{sec:eval:dtr}
\revb{We sweep $D_\mathrm{tr}$ from 4 to 9 ($\efflevel\!=\!9$) on \NAME (Fig.~\ref{fig:param-sweep:Dtr}) while dynamically applying BSGS to maximize performance in each configuration.
Raising $D_\mathrm{tr}$ from 4 to 7 adds 16\% more ModMuls but shrinks the CtS working set by 34\% and DRAM stalls by 77\%, cutting \textboot time from 1.18ms to 0.915ms, showing that fg-CtS enhances performance through optimizing memory use rather than reducing computation.
Setting $D_\mathrm{tr}$ higher at 8--9 (with reduced security) requires more \textevks, leading to more stalls and slower execution (0.981--1.04ms).
The gain does not come from tuning $D_\mathrm{tr}$ alone: using BSGS at $D_\mathrm{tr}\!=\!7$ (1.23ms) or not using BSGS at $D_\mathrm{tr}\!=\!4$ (1.36ms) is slower than the conventional $D_\mathrm{tr}\!=\!4$ with BSGS (1.18ms). FHE compilers~\cite{usenix-2024-dacapo, asplos-2025-halo, asplos-2025-resbm} can apply fg-CtS when $\efflevel$ suffices, and otherwise use a smaller $D_\mathrm{tr}$ with BSGS.}

\subsection{Alternative \NAME hardware configurations}
\label{sec:eval:alter}

\begin{figure*}
    \centering
    \subfloat[Number of clusters]{\includegraphics[width=0.36\columnwidth]{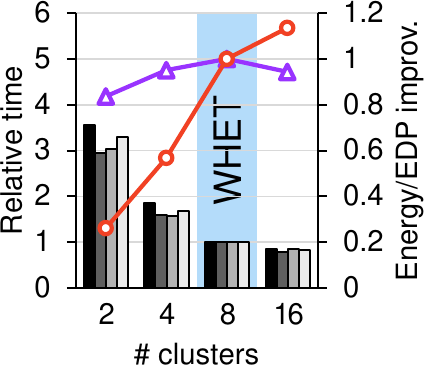}\label{fig:sweep:cluster}}
    \quad
    \subfloat[HBM bandwidth]{\includegraphics[width=0.515\columnwidth]{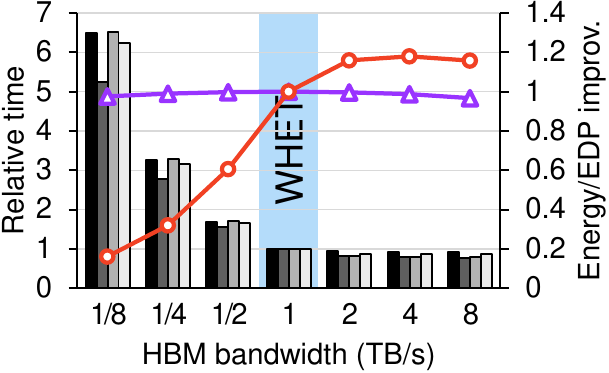}\label{fig:sweep:hbm}}
    \quad
    \subfloat[Main scratchpad size]{\includegraphics[width=1.03\columnwidth]{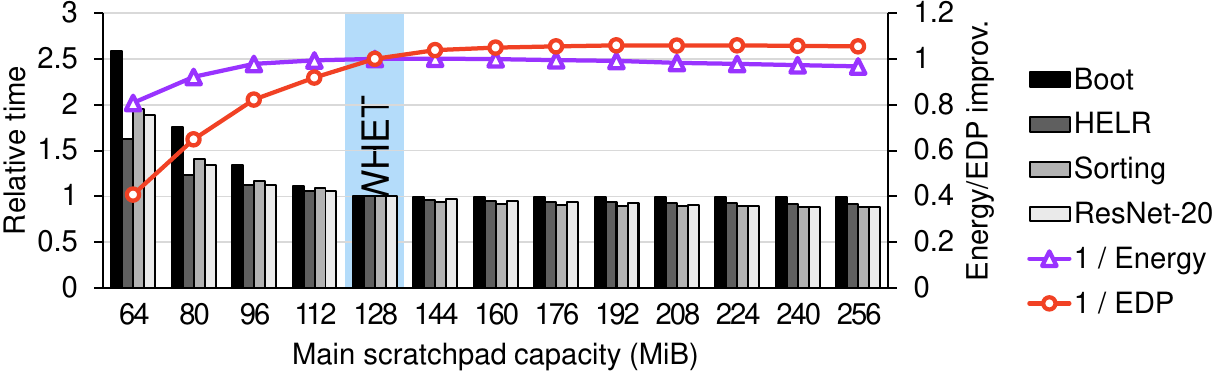}\label{fig:sweep:sram}}
    \vspace{-0.04in}
    \caption{Execution time, energy, and energy-delay product (EDP) under varying (a) cluster count (total on-chip memory capacity kept constant), (b) HBM bandwidth, and (c) main scratchpad capacity.}
    \label{fig:sweep}
    \vspace{-0.08in}
\end{figure*}

\revb{Fig.~\ref{fig:param-sweep:KM} partitions a fixed 160MiB SRAM budget between the main scratchpad and the KeyMult buffer and shows where the KeyMult outputs $(a_\mathrm{res}, b_\mathrm{res})$ (\S\ref{sec:arch:memory}) are stored and accessed.
Without the buffer, they induce 15.5GiB (42\% of the entire traffic) of main scratchpad accesses per \textboot.
A 32MiB buffer, the smallest that holds all of them, takes over all of these accesses and is fastest.
A smaller buffer spills some outputs back to the main scratchpad, and a larger one takes capacity away from it: with a 48MiB buffer, the main scratchpad shrinks to 112MiB and DRAM stalls rise from 9\% to 18\% of \textboot time.}

Fig.~\ref{fig:sweep} varies the number of clusters, HBM bandwidth, and main scratchpad capacity. Each shows diminishing returns around \NAME's values (eight, 1024GB/s, and 128MiB): doubling one improves performance by only 1.20$\times$, 1.16$\times$, and 1.09$\times$ (geometric mean), whereas halving one degrades it by 1.67$\times$, 1.65$\times$, and 1.99$\times$, again indicating a balanced design.

\noindent \textbf{Scaling \NAME.} Scaling \NAME requires increasing computational throughput and memory bandwidth together.
For scale-up, HBM bandwidth can be increased to several TB/s while adding more clusters to form a larger chip, similar to recent NVIDIA GPU designs.
For scale-out, prior multi-chip FHE acceleration efforts can be leveraged, such as CiFHER (custom ASIC chiplets)~\cite{seed-2024-cifher}, Cinnamon (custom ASICs)~\cite{asplos-2025-cinnamon}, Cerium (NVIDIA GPUs)~\cite{arxiv-2025-cerium}, and CROSS (Google TPUs)~\cite{hpca-2026-cross}.

\subsection{A wider range of CNN workloads}
\label{sec:eval:cnn}

\setlength{\tabcolsep}{4pt}
\begin{table}[t]
    \centering
    \caption{Evaluation of various CNN workloads based on the implementation of Orion~\cite{asplos-2025-orion}. SiLU activations are used.}
    \label{tab:eval-orion}
    \begin{tabular}{r|rrr|rr}
\toprule
  \multirow{2}{*}{Model} & \multicolumn{1}{c}{\BASE} & \multicolumn{1}{c}{\NAME} & \multicolumn{1}{c|}{\NAME} & \multicolumn{2}{c}{CPU (1-thread)} \\
  & \multicolumn{1}{c}{(180MiB)} & \multicolumn{1}{c}{(128MiB)} & \multicolumn{1}{c|}{(180MiB)} & Orion\textsuperscript{$\dagger$} & Plain\textsuperscript{$\ddagger$}\\
\midrule
  VGG-16 & 92.6ms & 94.9ms & 76.6ms &  397s & 11.9ms\\
  (CIFAR-10) & (5.32J)& (5.16J)& (4.68J)\\
\midrule
  ResNet-20 & 50.8ms & 33.1ms & 28.8ms&  301s & 2.30ms\\
  (CIFAR-10) & (3.93J)& (3.15J)& (3.18J)\\
\midrule
  MobileNet & 159ms & 98.5ms&  87.0ms & 892s & 3.12ms \\
  (Tiny ImNet) & (12.2J) & (9.55J)& (9.45J)\\
\midrule
       ResNet-18 & 304ms & 239ms& 216ms& 1447s & 45.6ms\\
       (Tiny ImNet) & (21.9J)& (18.7J)& (18.4J)\\
\bottomrule
    \end{tabular}
    \begin{itemize}[leftmargin=*]
    \item[$\dagger$] Reported from the Orion paper~\cite{asplos-2025-orion} on Intel Xeon Platinum 8581C.
    \item[$\ddagger$] Unencrypted inference with PyTorch 2~\cite{asplos-2024-pytorch2.0} on Intel Xeon Platinum 8358.
    \end{itemize}
    \vspace{-0.07in}
\end{table}
\setlength{\tabcolsep}{6pt}

\begin{figure}[t]
    % \color{blue}
    \centering
    \includegraphics[width=0.98\columnwidth]{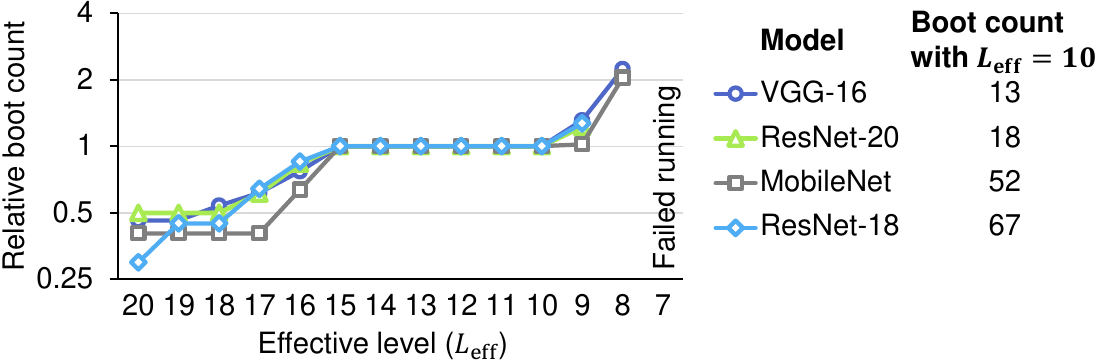}
    \vspace{-0.06in}
    \caption{\revb{\textboot count vs.\ \textefflevel for Orion CNNs~\cite{asplos-2025-orion}, normalized at $\efflevel\!=\!10$.}}
    \label{fig:bootcount}
    \vspace{-0.1in}
\end{figure}

We evaluate \NAME using a more diverse set of CNN models implemented with the state-of-the-art FHE framework Orion~\cite{asplos-2025-orion}.
We follow Orion’s implementation faithfully, incorporating Min-KS to reduce \textevk usage.
\NAME is evaluated with two main scratchpad configurations: the default 128MiB and an expanded 180MiB.
\revb{We used $\efflevel\!=\!10$ for both \NAME and \BASE with the best parameter tuning for each.
Orion shows a typical case where additional levels may not be useful: Fig.~\ref{fig:bootcount} shows that \textboot count stays the same for \textefflevel of 10--15 for all the models.}
The results are shown in Table~\ref{tab:eval-orion}.

Across ResNet-20~\cite{cvpr-2016-resnet} (CIFAR-10~\cite{techreport-2009-cifar10}), MobileNet~\cite{arxiv-2017-mobilenets} (Tiny ImageNet~\cite{tiny-imagenet}), and ResNet-18 (Tiny ImageNet), the default 128MiB \NAME achieves 1.27--1.61$\times$ speedups and 1.17--1.27$\times$ energy reductions over \BASE.
\revb{For ResNet-20, \NAME and \BASE spend 19.3ms and 38.2ms, respectively, on the 18 \textboot computations.
Increasing \textefflevel for \BASE does not help; e.g., $\efflevel\!=\!13$ takes 56.5ms for ResNet-20 instead of 50.8ms ($\efflevel\!=\!10$) with the same \textboot count.}

VGG-16~\cite{iclr-2015-vgg} (CIFAR-10) presents a different behavior.
Orion’s implementation involves extremely large linear transforms \revb{outside \textboot}, making BSGS unavoidable and requiring storage of \revb{numerous} baby-step ciphertexts.
\revb{While \NAME is still 2.2$\times$ faster than \BASE (13.0 vs.\ 28.4ms) for \textboot computations in VGG-16, fg-CtS does not touch the rest of the computations (mostly linear transforms), and}
the reduced 128MiB scratchpad becomes a bottleneck.
Increasing the scratchpad to 180MiB restores performance, yielding a 1.24$\times$ speedup for \NAME.
\revb{Alternatively, capping the number of baby ciphertexts ($\text{\#baby}\le16$) on both systems---following capacity-aware BSGS tuning~\cite{isca-2023-sharp, asplos-2026-falcon}---improves overall latency for both systems; under this configuration, 128MiB \NAME (76.2ms) outperforms \BASE (89.7ms).}

Overall, \NAME substantially narrows the gap between encrypted and unencrypted inference.
Orion reports over a 30,000$\times$ slowdown on a single-threaded CPU compared to plaintext inference.
\NAME reduces this gap to 4.74--31.6$\times$, bringing FHE closer to practical use.

\noindent\textbf{Capacity limitations.} Larger CNNs require improved execution strategies for accelerators.
Even with HBM4~\cite{jedec-2025-hbm4}, two HBM stacks provide at most 128GiB of capacity, whereas Orion materializes hundreds of GiBs of plaintexts for large inputs such as ImageNet~\cite{cvpr-2009-imagenet}.
FHE CNN designs~\cite{ccs-2024-neujeans, tifs-2023-coeff, access-2024-hyphen} using fewer plaintexts suit such accelerators.

\section{Related Work}
\label{sec:related}

Previous GPU studies have attempted to enhance the utilization of GPU hardware resources for FHE.
TensorFHE~\cite{hpca-2023-tensorfhe} utilizes tensor cores in recent NVIDIA GPUs to accelerate NTT.
WarpDrive~\cite{hpca-2025-warpdrive} combines the use of tensor cores and regular CUDA cores.
%concurrently on Tensor Cores and CUDA cores, and
NEO~\cite{isca-2025-neo} and HyperDrive~\cite{isca-2026-hyperdrive} redirect HE kernels to idle FP64 pipelines.
Jung et al.~\cite{tches-2021-100x}, Cheddar~\cite{asplos-2026-cheddar}, and Helios~\cite{micro-2026-helios} apply kernel fusion to improve the memory bandwidth utilization.
However, limited on-chip memory capacity and bandwidth in GPUs remain bottlenecks for FHE.

Numerous FPGA solutions have also been proposed.
Roy et al.~\cite{hpca-2019-roy} and HEAX~\cite{asplos-2020-heax} deploy spatial pipelines for the core primitives in HE without support for \textboot.
FAB~\cite{hpca-2023-fab} and Poseidon~\cite{hpca-2023-poseidon} improve upon these studies and support full-slot CKKS \textboot.
Hydra~\cite{hpca-2025-hydra} takes a scale-out approach with multiple FPGAs to support large private ML workloads.
EFFACT~\cite{hpca-2025-effact} provides a full-stack platform with a compiler and a streaming memory controller for both FPGA and ASIC.

Several processing-in-memory (PIM) and near-data processing (NDP) architectures target memory-intensive HE.
Anaheim~\cite{hpca-2025-anaheim} offloads memory-bound operations to PIM while keeping compute-heavy kernels on GPUs. 
FHENDI~\cite{hpca-2025-fhendi} leverages an HBM-based NDP design for large FHE workloads.
HE\textsuperscript{2}~\cite{isca-2026-he2} proposes a communication-light ASIC-NMP heterogeneous architecture.
CIPHERMATCH~\cite{asplos-2025-ciphermatch} uses in-flash processing for bulk homomorphic additions.

F1~\cite{micro-2021-f1} is the first ASIC FHE accelerator but provides limited support for CKKS \textboot with $N\!=\!2^{14}$.
BTS~\cite{isca-2022-bts}, fully supporting CKKS \textboot with $N\!=\!2^{17}$, deploys 2,048 small processing units connected by a two-dimensional NoC.
CraterLake~\cite{isca-2022-craterlake} realizes a \textboot-oriented design by creating a large cluster for $N\!=\!2^{16}$ based on F1.
ARK~\cite{micro-2022-ark} develops algorithmic optimizations, such as Min-KS, to mitigate off-chip memory bandwidth bottlenecks.
SHARP~\cite{isca-2023-sharp} explores the impact of word size and suggests using 36 bits for FHE accelerators.
UFC~\cite{micro-2024-ufc}, Trinity~\cite{micro-2024-trinity}, and UniFHE~\cite{hpca-2026-unifhe} modify the architecture to support both RLWE-based (\eg CKKS) and non-RLWE-based (\eg TFHE~\cite{jc-2020-tfhe}) FHE schemes.
%and reduces the working set size by BSGS fine-tuning.
FAST~\cite{isca-2025-fast} and HAWK~\cite{micro-2025-hawk} exploit a new $\mathtt{KS}$ method, KLSS~\cite{crypto-2023-klss}, and add hardware support for KLSS on top of SHARP\revb{; both keep BSGS for CtS and size their on-chip memory (281MiB and 212MiB) to that working set, which \NAME shrinks. KLSS changes only how $\mathtt{KS}$ is computed, so \NAME still applies}.
Falcon~\cite{asplos-2026-falcon} shares similar motivations with \NAME, exploring memory-capacity-aware algorithm modifications\revb{: it cuts BSGS computation and sizes baby steps to the on-chip memory, whereas \NAME removes BSGS from CtS, and Falcon's tuning applies to the application linear transforms that still keep BSGS (\S\ref{sec:eval:cnn})}.
AutoFHE~\cite{isca-2026-autofhe} automatically generates domain-specific FHE accelerators from RTL designs.
Intel recently taped out HERACLES, a 197mm\textsuperscript{2} ASIC FHE accelerator in the Intel 3 node~\cite{isscc-2026-heracles}.

Finally, there have been attempts to better adapt CKKS to hardware via efficient scheduling (MAD~\cite{micro-2023-mad}, CROPHE~\cite{hpca-2026-crophe}), cryptographic modification (BitPacker~\cite{asplos-2024-bitpacker}, Grafting~\cite{ccs-2025-grafting}), \textboot-aware compiler development (DaCapo~\cite{usenix-2024-dacapo}, HALO~\cite{asplos-2025-halo}, ReSBM~\cite{asplos-2025-resbm}), or ciphertext packing (FEnc$^2$~\cite{isca-2026-fenc2}, OptiPrime~\cite{micro-2026-optiprime}).
Also, several studies~\cite{seed-2024-cifher, asplos-2025-cinnamon, arxiv-2025-cerium, hpca-2026-cross} explore multi-chip scaling for FHE acceleration (\S\ref{sec:eval:alter}).
\section{Conclusion}
\label{sec:conclusion}

We presented \NAME, which incorporates cryptographic and algorithmic techniques to better align FHE execution with accelerator architectures, along with lightweight architectural refinements to improve hardware utilization.
By analyzing the working set at the highest levels, \NAME introduces fg-CtS, plaintext compression, and intermediate ModRaise to substantially reduce data footprint and off-chip memory traffic during CKKS \textboot.
Additional architectural enhancements, including the KeyMult buffer and extended instructions for element-wise operations, further alleviate on-chip memory bandwidth and data dependency bottlenecks.
Overall, \NAME achieves 1.32--\revb{2.15}$\times$ speedups over our controlled baseline, \BASE, and enables fast private CNN inference with 33.1--239ms latencies on representative models.

\section*{Acknowledgment}
\label{sec:acknowledgment}

This research was in part supported by Institute of Information \& communications Technology Planning \& Evaluation (IITP) grant funded by the Korea government (MSIT) [RS-2021-II211343, RS-2023-00256081, RS-2025-02217656].
The EDA tool was supported by the IC Design Education Center (IDEC), Korea.
Hyesung Ji is with the Interdisciplinary Program in Artificial Intelligence (IPAI), Seoul National University (SNU).
Wonseok Choi and Hyunah Yu are with the Department of Intelligence and Information, SNU.
Jung Ho Ahn, the corresponding author, is with the Department of Intelligence and Information and IPAI, SNU.

%%%%%%% -- PAPER CONTENT ENDS -- %%%%%%%%

%%%%%%%%% -- BIB STYLE AND FILE -- %%%%%%%%
\balance
\bibliographystyle{IEEEtranS}
\bibliography{refs_with_url}
%%%%%%%%%%%%%%%%%%%%%%%%%%%%%%%%%%%%
\appendix
\subsection{AI use}
In accordance with the IEEE Publication Services and Products Board Operations Manual (Section 8.2.1.B.10), the authors disclose the use of ChatGPT, Codex, and Claude for language refinement and experimental support, including code modification, experiment execution, and data analysis. All cryptographic, algorithmic, and architectural contributions are entirely the authors’ own, and all AI-assisted output was reviewed and verified by the authors.

\section{Symbols}
\label{app:symbol}

\end{document}